\documentclass[twoside]{article}

\usepackage[accepted]{aistats2021}
\usepackage[T1]{fontenc}    
\usepackage{hyperref}       
\usepackage{url}            
\usepackage{amsfonts}       
\usepackage{microtype}      
\usepackage{lipsum}

\usepackage{amssymb}
\usepackage{amsmath}
\usepackage{amsthm}
\usepackage{xcolor}
\usepackage{xfrac}
\usepackage{bbm}

\usepackage{natbib}
\usepackage{booktabs}
\usepackage{enumitem}

\newtheorem{assumption}{Assumption}

\newtheorem{theorem}{Theorem}

\newtheorem{corollary}{Corollary}
\newtheorem{lemma}{Lemma}

\newtheorem{definition}{Definition}
\usepackage{apptools}
\AtAppendix{\counterwithin{lemma}{section}}
\AtAppendix{\counterwithin{theorem}{section}}
\AtAppendix{\counterwithin{corollary}{section}}

\makeatletter 
\def\blfootnote{\xdef\@thefnmark{}\@footnotetext}
\makeatother

\newcommand\var{\text{var}}

\newcommand{\norm}[1]{\left\lVert#1\right\rVert}

\newcommand{\X}{\mathbf{X}}
\newcommand{\SX}{\tilde{\X}}

\newcommand{\cc}{\mathbf{c}}
\newcommand{\CC}{\mathbf{C}}
\newcommand{\scc}{\tilde{\cc}}
\newcommand{\SC}{\tilde{\CC}}
\newcommand{\TC}{\tilde{C}}
\newcommand{\SL}{\mathbf{S}}

\newcommand{\x}{\mathbf{x}}
\newcommand{\ww}{\mathbf{w}}
\newcommand{\y}{\mathbf{y}}
\newcommand{\A}{\mathbf{A}}

\newcommand{\Y}{\mathbf{Y}}

\newcommand{\Z}{\mathbf{Z}}
\newcommand{\EE}{\mathbb{E}}
\newcommand{\E}{\mathcal{E}}
\newcommand{\PP}{\mathbb{P}}

\newcommand{\e}{\mathbf{e}}
\newcommand{\aaa}{\mathbf{a}}
\newcommand{\I}{\mathbf{I}}
\newcommand{\1}{\mathbf{1}}
\newcommand{\z}{\mathbf{z}}

\newcommand{\s}{\mathbf{s}}

\newcommand{\SE}{\tilde{\E}}
\newcommand{\SBE}{\tilde{\boldsymbol{\E}}}

\newcommand{\sx}{\tilde{\x}}
\newcommand{\sy}{\tilde{\y}}

\newcommand{\htau}{\hat{\tau}}
\newcommand{\hbeta}{\hat{\beta}}
\newcommand{\tbeta}{\tilde{\beta}}
\newcommand{\se}{\tilde{\boldsymbol{\epsilon}}}
\newcommand{\sepsilon}{\tilde{\epsilon}}

\newcommand{\sdelta}{\tilde{\delta}}
\newcommand{\sbdelta}{\tilde{\boldsymbol{\delta}}}

\usepackage{algorithmic}

\begin{document}
\twocolumn[
\aistatstitle{Designing Transportable Experiments Under S-Admissability}

\aistatsauthor{
My Phan\And David Arbour* \And
 Drew Dimmery* \And Anup B. Rao*  
}
\aistatsaddress{ University of Massachusetts\\ Amherst \\ \And  Adobe Research \And Facebook Core Data Science \And Adobe Research } ]

\begin{abstract}
We consider the problem of designing a randomized experiment on a source population to estimate the Average Treatment Effect (ATE) on a target population. We propose a novel approach which explicitly considers the target when designing the experiment on the source. Under the covariate shift assumption, we design an unbiased importance-weighted estimator for the target population's ATE. To reduce the variance of our estimator, we design a covariate balance condition (Target Balance) between the treatment and control groups based on the target population. We show that Target Balance achieves a higher variance reduction asymptotically than methods that do not consider the target population during the design phase.  Our experiments illustrate that Target Balance reduces the variance even for small sample sizes. 

\end{abstract}

\section{Introduction}
The problem of generalization is present everywhere that experiments are run.
In the online environment, tests are run with the users who show up on the product while the experiment is running (and are therefore highly active users), while inferences about user experience are most useful on the full set of users (both highly active and less active)~\citep{wang2019heavy}.
In clinical research, it is an omnipresent problem to recruit minorities into randomized trials~\citep{fisher2011challenging}, thus making it difficult to assume that the measured effects will generalize to the larger population of interest (e.g. the United States as a whole, or people afflicted with a particular health condition).
In lab experiments, the sample is often one of convenience such as undergraduates in rich countries or from pools of potential subjects available online~\citep{henrich2010most}. 
Field experiments in governance or development such as \citet{dunning2019voter} are conducted in particular countries or in particular communities, but the policy implications of such work stretch far beyond the borders of the study population.
As in \citet{dunning2019voter}, the desire is not just to understand how Burkina Faso voters respond to more information about their political leaders, but to understand how voters across the world might respond to similar informational treatments.
The same is true for experiments in development economics, such as microfinance~\citep{meager2019understanding} and for studies of internet phenomena~\citep{munger2018temporal}.
In these cases, it isn't a \emph{surprise} after running an experiment that generalizing the knowledge is important; indeed, generalization of knowledge to a broader population is core to the motivation for the experiment in the first place.

We pre-suppose that an experimenter knows ex-ante the population on which they wish to draw broader inferences.
The task we consider, therefore, is to design an experiment that best allows the generation of causal knowledge on this inferential target.
While previous work~\citep{hartman2015sate, dehejia2019local, stuart2011propensity, dugoff2014generalizing} has examined corrections on the analysis side to extrapolate estimates from sample to target population, the novelty of this work is in doing this through a \emph{design-based} solution.
That is, if you know your goal is to generalize to a target population, we consider how that should modify experimental design.
We focus in particular on the ``S-admissability'' condition for transportability, in which the outcome distribution conditional on a set of covariates is the same in both the source and target distributions~\citep{pearl2011transportability}.
\blfootnote{*Equal contribution.}

\textbf{Contributions. } Using the Mahalanobis distance and importance weighting, we design an estimator with a balancing condition for the target distribution's ATE that is unbiased and has low variance. 
\begin{itemize}[nosep,leftmargin=1em,labelwidth=*,align=left] 
    \item In Section~\ref{sec:design}, we introduce an importance-weighted estimator with a balance condition called Target Balance that explicitly considers the target distribution in the design phase. 
    \item In Section~\ref{sec:expectation} we show that using the importance-weighted estimator with Target Balance results in an unbiased estimator of the target distribution's ATE (Theorem~\ref{thm:unbiased})
    \item We analyze the variance assuming a linear model. In Section~\ref{sec:d1}, we show that when the dimension of the covariates $d=1$, for a finite sample size $n$, Target Balance reduces the variance (Corollary~\ref{col:var1d}). Moreover, among all balance criteria with rejection probability at most $\alpha$ (including balancing by only considering the source distribution, which we call Source Balance), Target Balance achieves the optimal variance reduction (Theorem~\ref{thm:optimality_1d}). When $d \ge 1$ (Section~\ref{sec:asymptotic}), when the sample size is large, Target Balance reduces the variance (Theorem~\ref{thm:var_reduction}) and achieves a lower variance than Source Balance (Theorem~\ref{thm:asymptotic_optimal}). 
    \item In Section~\ref{sec:simulation} we perform experiments\footnote{Code for this paper is available at~\url{https://github.com/myphan9/Designing_Transportable_Experiments}} to show that Target Balance has small mean-squared errors even for $d > 1$, small sample size and non-linear model. 
\end{itemize}

\section{Problem Setting}
\label{sec:problem_setting}
We first fix notation before proceeding to the problem setting. 
Upper-case letters are used to denote random variables, lower-case letters are used to denote values taken by them. We use bold-faced letters to denote $n$ samples and normal letters to denote a single sample. For example, $X_i \in \mathcal{R}^d$ is a random variable denoting the covariates of sample $i$. $\X = (X_1, ..., X_n)^T \in \mathcal{R}^{n \times d}$ is the random variables $X_1, .., X_n$ concatenated together. $x_i \in \mathcal{R}^d$ is a value of $X_i$, and $\x = (x_1, ..., x_n)^T \in \mathcal{R}^{n \times d}$ is a value of $\X$. 

Some random variables, like $X$, can have two different distributions, either source distribution or target distribution. In that case, we use $\EE^S X$, $\var^S X$ and $cov^S X$ to denote the expectation, variance and covariance with respect to the source distribution, and $\EE^T X$, $\var^T X$ and $cov^T X$ to denote with respect to the target distribution. We use no superscripts when there is no confusion. For example $\EE A_i$ is the expectation of the treatment assignment $A_i$ of sample $i$. For a random variable $R$, we use $\EE_R$, $\var_R$ and $cov_R$ to denote the expectation, variance and covariance over the randomness of $R$. For example, $\EE^S_{\X}$ denote the expectation over the randomness of $\X = (X_1, ..., X_n)^T$ according to the source distribution. We omit the  subscripts when it's clear. 

The problem considered in this paper is as follows. 
We assume that we are presented with two populations, referred to as the source and target populations, with corresponding densities $p_S$ and $p_T$, respectively.
We further assume that we observe a set of pre-treatment covariates from the source population, $x_1, \dots x_n \sim p_S$.
We assume that we are freely able to assign treatment, $a_1, \dots, a_{n} \in \{0, 1\}$ to individuals observed in the source population and observe their outcomes, $y_1, \dots, y_{n} \in \mathcal{R}$.
The estimand of interest is the average treatment effect for the \textit{target} population~(the population of individuals which were not subject to an experiment), 
\begin{align}
\label{eq:transportATE}
    \tau^T_Y =  \mathbb{E}^T [Y^{A=1} - Y^{A=0}] .
\end{align}
Where $Y^{A=0}$, $Y^{A=1}$ are the potential outcomes~\citep{rubin2011causal},  i.e., the values of $Y$ that would have been observed had treatment been observed at $A=1$ or $A=0$, respectively. We use $Y$ to denote $(Y^0, Y^1)$ and $Y^*$ to denote the \emph{observed} outcome.

In order to make this problem tractable we will assume the following throughout the remainder of the paper:
    \begin{assumption} 
    \label{ass:source_target}
    Equality of conditional densities, i.e., 
    $p_S(Y | X) = p_T(Y|X)$
    (note $p_S(X) \neq p_T(X)$ in general).
    \end{assumption}
This assumption places identification of the transportability of effects under the rubric of S-admissability~\citep{pearl2011transportability}.
    \begin{assumption}
    Overlap between source and target distributions, i.e., 
    $p_T(X) > 0 \implies p_S(X) > 0$.
    \end{assumption}
    \begin{assumption}
    \begin{align*}
        Y^1 = \psi(X)^T \beta_1 + \E_1 & \quad \quad
        Y^0 =  \psi(X)^T \beta_0 + \E_0
        \end{align*}
        where $\psi$ is a basis function and $\E_1, \E_0$ are mean zero random variables.
    \end{assumption}
To reduce notational clutter, and without loss of generality, we will assume that $\psi$ is the identity function for the remainder of paper so that we can write $X$ instead of $\psi(X)$. 
    \begin{assumption} 
    The ratio of the pdfs, $p_T(X)/p_S(X)$, is known. 
    \end{assumption}
In a nested trial design in which the sampling probabilities from the population are known \citep{dahabreh2019study}, this ratio can be calculated as $$\frac{p_T(X)}{p_S(X)}=\frac{p(X|S=0)}{p(X|S=1)} = \frac{p(S=0|X)}{p(S=1|X)}\frac{p(S=1)}{p(S=0)},$$ where $S=1$ indicates that the unit is selected to be in the source and $S=0$ indicates that the unit is not selected and is in the target. 

The assumptions, though nontrivial,  are common throughout the literature on transportability~\citep{stuart2011propensity,hartman2015sate,pearl2011transportability}.
We conjecture that similar results to those in this paper will hold in the case in which importance weights are estimated with parametric convergence rates.
We leave this extension as future work.

For a sample $i$, let $X_i, A_i$, $Y_i^a$ and $Y^*_i$ be the covariates, treatment, outcome of treatment $a$, and observed outcome. Let $n_0$ be the size of the control group (where $A_i=0$) and $n_1$ be the size of the treatment group (where $A_i=1$).  Similar to common practice~(c.f., \citet{stuart2011propensity,hartman2015sate,rudolph2017robust, buchanan2018generalizing}), we infer $\tau^T_Y$ with importance weights, 
\begin{align}
\label{eq:importance_sample}
\hat{\tau}^T_Y &:= \frac{1}{n_1} \sum_{A_i=1} W_i Y^{*}_i - \frac{1}{n_0} \sum_{A_i=0}^n W_i Y^{*}_i \nonumber \\
   &= \frac{1}{n_1} \sum_{i=1}^n W_iA_i Y^1_i - \frac{1}{n_0} \sum_{i=1}^n W_i (1-A_i) Y^0_i
\end{align}
where $W_i = \frac{p_T(X_i)}{p_S(X_i)} $.
While equation \ref{eq:importance_sample} is unbiased, the estimate can incur large variance in the presence of large importance weights.

For ease of notation we define $Z_i = 2A_i - 1 \in \{-1, 1\}$ and  let $\Z$ be the $n \times 1$ vector of random variables $Z_1, ..., Z_n$ and $\z$ be a value taken by $\Z$. $Y_i = (Y^0_i, Y^1_i)$ is a random variable denoting all possible outcomes of sample $i$. $\Y = (Y_1, ..., Y_n)^T \in \mathcal{R}^{n \times 2}$ is the random variables $Y_1, .., Y_n$ concatenated together. $y_i \in \mathcal{R}^2$ is a value of $Y_i$, and $\y = (y_1, ..., y_n)^T \in \mathcal{R}^{n \times 2}$ is a value of $\Y$. Let $w_i  = \frac{p_T(x_i)}{p_S(x_i)} $ and $\ww$ be the $n \times n$ diagonal matrix with $\ww(i,i) = w_i.$  For matrix $\aaa$, we use $\tilde{\aaa}$ to denote $\ww\aaa$ where each row $i$ of $\aaa$ is multiplied by $w_i$.
\section{Designing for Transportation}
\label{sec:design}


We consider $n_0 = n_1 = n/2$ throughout the paper. The core contribution of this work is a procedure to estimate equation \ref{eq:transportATE} which explicitly considers the target population when designing the experiment for the source population. 
We focus on adapting re-randomization, an experimental design procedure which optimizes balance, i.e., the difference in means of $X$ between treatment and control groups.
Specifically, rerandomization centers on a \textit{balance criterion},
\begin{definition}[Target Balance]
\label{def:target_balance}
With a rejection threshold $\alpha$, define the balance condition:
\begin{align*}
\phi_T^{\alpha} ( \x, \Z) =
    \begin{cases}
      1, & \text{if}\ M(\frac{2}{n} (\ww\x)^T\Z) < a(\x) \\
      0, & \text{otherwise}
    \end{cases}
\end{align*}
where $M(\frac{2}{n} (\ww\x)^T\Z)$ is a distance (defined below in Eq.~\ref{eq:M_def}) between the covariates associated with treatment and control given by $\Z$ and $a(\x)$ is chosen such that $\PP(\phi_T^{\alpha} = 1| \x) = 1- \alpha$. 
\end{definition}
We omit $\alpha$ and simply write $\phi_T$ when $\alpha$ is not necessary for exposition. We omit $\x$ and write $a$ when there is no confusion. 

The full assignment procedure is then
\begin{enumerate}[nosep,leftmargin=1em,labelwidth=*,align=left]
    \item Assign $A$ randomly for each person $1, \dots, n$ such that $\sum_i a_i = n_1$. There are $\binom{n}{n_1}$ ways to choose this, each of which is equally likely. 
    \item If $\phi_T(\x,\z) = 0$ return to step (1). 
    \item Conduct experiment with treatment assignments, $A$. 
\end{enumerate}

Following standard practice in rerandomization~\citep{morgan2012rerandomization}, we  will focus on a criterion based on Mahalanobis distance, but incorporating a weighting term to express our desire for balance in the \emph{target} distribution rather than in the source.
We refer to this weighted Mahalanobis distance as $M(\frac{2}{n} (\ww\x)^T\Z)$, where $M(\cdot)$ is defined as:
\begin{align}
\label{eq:M_def}
    \nonumber M(U)&:= 
    \left(U \right)^T Cov(U)^{-1}
     \left(U \right)\\
     &= \|B\|^2 \text{ where } B = U Cov(U)^{-1/2} 
\end{align}
Thus, the balance condition $M(\frac{2}{n} (\ww\x)^T\Z) < a$ is equivalent to truncating the square norm of $B$ to be less than $a$.
Note that this is a standardized measure of the difference in importance-weighted covariate-means between treatment and control, since
\begin{align*}
    \frac{2}{n} \sum_{Z_i = 1} w_i x_i - \frac{2}{n} \sum_{Z_i = -1} w_i x_i = \frac{2}{n} (\ww\x)^T\Z.
\end{align*}
Thus, rerandomization simply rejects designs with covariate imbalance larger than a pre-specified value.

The novelty in our proposed design is to reject samples based on imbalance in the the \emph{target} distribution rather than based on imbalance in the \emph{source} distribution. The standard in the rerandomization literature is to focus on balance in the \emph{source} distribution, which in our setup implies assuming that the target distribution is equal to the source distribution.
Therefore, importance weights in this case are all equal to one.
We call this balancing condition \textit{Source Balance}, which we denote by $\phi^{\alpha}_S(\x,\Z)$.

We explain the intuition behind using Target Balance rather than Source Balance. An unbiased estimator for the source's ATE is:
$ \hat{\tau}^S_Y := \frac{1}{n_1}\sum_i A_i Y_i^1 - \frac{1}{n_0}\sum_i A_i Y_i^0 .$ There are existing results \citep{li2018asymptotic, harshaw2019balancing} that can be applied to linear models to show variance reduction of $\hat{\tau}^S_Y$ with Source Balance defined by $\x$. 

By defining new variables $\tilde{Y}^a_i = W_i \cdot Y^a_i$ for $a \in \{0,1\}$, the importance weighted estimator can now be expressed in a similar form to $\hat{\tau}^S_Y$: $ \hat{\tau}^T_Y = \frac{1}{n_1}\sum_i A_i \tilde{Y}_i^1 - \frac{1}{n_0}\sum_i A_i \tilde{Y}_i^0 .$ 

We are no longer in a linear setting because
 $\tilde{Y}^a_i =  W_i ( \beta_a^T X_i  + \E_a).$ But by defining $\tilde{X}_i = W_i X_i$ and $\tilde{\E} =W_i \E$ we have:
$\tilde{Y}^a_i =  \beta_a^T \tilde{X}_i + \tilde{\E}_a,$
where  $\EE \left[ \tilde{Y}^a | X \right]$ is a linear function of $\tilde{X}$, which can be considered as a feature-transformed $X.$ \citep{li2018asymptotic, harshaw2019balancing} can now be applied to estimate $\hat{\tau}^T_Y$ with Target Balance defined by $\tilde{\x} = \ww \x$. 

Let $\rho(\x,\z) \in \{0,1\}$ be a function of $\x$ and $\z$ used in the re-randomization procedure. Note that since we re-sample $\Z$ until $\rho=1$, only the distribution of $\Z$ is affected by the balance condition. Let $\Z_{\rho}$ denote the distribution of $\Z$ after being accepted by the balance condition $\rho=1$. 
\section{Related Work}
\label{sec:related_work}
Our work relates to and ties together two distinct strands of research: (1) ex-post generalization of experimental results to population average effects and (2) ex-ante experimental design.
We will discuss each in turn.

\textbf{Generalization.}

Within the literature on methods for generalization, work has generally focused on ex-post adjustments to experiments previously run.

The foundational work of \citet{stuart2011propensity} provides an approach based on propensity scores for generalizing the results of experimental interventions to target populations.
Our work will leverage this general framework, but introduce methods for optimizing an experimental design to ensure effective generalization performance of resulting estimates.
\citet{hartman2015sate} similarly uses a combination of matching and weighting to generalize experimental results in-sample to a population average treatment effect on the treated.
Other work has also considered weighting-based approaches to generalization~\citep{buchanan2018generalizing}.

\citet{dehejia2019local} shows how to use an outcome-modeling approach to extrapolate effects estimated in one population to a population.
In contrast to \citet{hartman2015sate} and \citet{stuart2011propensity}, this approach relies on modeling the outcomes and then predicting effects in different locations rather than simply reweighting data observed in-sample.

\citet{dahabreh2018extending} provides a variety of estimation methods to generalize to a target population, including doubly-robust methods.
\citet{rudolph2017robust}, likewise, provides a doubly-robust targeted maximum likelihood estimator for transporting effects.

There has also been work focused particularly on identification in this setting. 
\citet{dahabreh2019study} defines a rigorous sampling framework for describing generalizability of experimental results and identifiability conditions through the g-formula.
\citet{pearl2011transportability} lays out a general framework for determining identifiability of effects generalized to new populations through.

\citet{miratrix2018worth} and \citet{coppock2018generalizability} challenge the premise of the necessity for generalization due to the rarity of heterogeneous treatment effects.
These studies specifically focused on survey experiments, however, and it isn't truly up for debate that many important objects of study have important heterogeneous components~\citep{allcott2015site,vivalt2015heterogeneous,dehejia2019local}.

\textbf{Experimental design.}

The standard practice for experimental design is blocking~\citep{greevy2004optimal}, in which units are divided into clusters and then a fixed number of units within each cluster are assigned to treatment.
This ensures balance on the cluster indicators within the sample.
\citet{higgins2016improving} provides a blocking scheme based on k-nearest-neighbors that can be calculated more efficiently than the ``optimal'' blocking of~\citep{greevy2004optimal}.

\citet{kallus2018optimal} takes an optimization approach to the problem of experimental design. This work optimizes treatment allocations based on in-sample measures of balance (particularly with respect to kernel means), showing how assumptions of smoothness are necessary to improve on simple Bernoulli randomization.

Rerandomization approaches simply draw allocations randomly until one is located which meets the pre-specified balance criteria.
This is also the basis of our proposed method.
\citet{morgan2012rerandomization} analyzes the rerandomization procedure of discarding randomized assignments that have more in-sample imbalance than a pre-specified criteria in terms of Mahalanobis distance.
\citet{li2018asymptotic} provides asymptotic results for rerandomization that does not rely on distributional assumptions on the covariates.

\citet{harshaw2019balancing} provides an efficient method for obtaining linear balance using a Gram-Schmidt walk.
Their algorithm includes a robustness-balance tradeoff tuneable by a parameter in their algorithm, and provides useful tools for analyzing experimental design which we use in our theoretical analyses in Section~\ref{sec:analysis}.

All aforementioned work on experimental design places as its objective estimation of effects on the sample (i.e. it optimizes for the sample average treatment effect).
This work departs by considering the alternative objective of prioritizing estimation on a target population (i.e. the population average treatment effect).



\section{Analysis}
\label{sec:analysis}
In this section we  will analyze the expectation and variance of our importance-weighted estimator in Eq.~\ref{eq:importance_sample} with Target Balance in Definition~\ref{def:target_balance}. 

Section~\ref{sec:expectation} shows that using the importance-weighted estimator with Target Balance results in an unbiased estimator of the target's ATE (Theorem~\ref{thm:unbiased}). 

In Section~\ref{sec:variance} we analyze the variance. In Section~\ref{sec:d1}, Corollary~\ref{col:var1d} shows that when the dimension of the covariates $d=1$, for a finite sample size $n$, Target Balance reduces the variance. Moreover, among all reasonable balance criteria with rejection probability at most $\alpha$ (including Source Balance), Target Balance achieves the optimal variance reduction (Theorem~\ref{thm:optimality_1d}). Section~\ref{sec:asymptotic} shows that when $d \ge 1$, when the sample size is large, Target Balance reduces the variance (Theorem~\ref{thm:var_reduction}) and achieves a lower variance than Source Balance (Theorem~\ref{thm:asymptotic_optimal}). 
\subsection{Expectation}
\label{sec:expectation}
In this section we will show that our importance-weighted estimator in Eq.~\ref{eq:importance_sample} is an unbiased estimator of the target's ATE with Target Balance: 

\begin{theorem}
\label{thm:unbiased}
Let $\hat{\tau}^T_Y $ be the importance-weighted estimator in Equation~\ref{eq:importance_sample}. When $n_0 = n_1 = n/2$:
$
     \mathbb{E}^S_{\X, \Y, \Z_{\phi_T}} \left[\hat{\tau}^T_Y   \right] = \tau^T_Y.
$
\end{theorem}

The proof makes use of the fact that the conditional distributions of $Y$ given $X$ in both the source and the target are the same ($p_S(Y|X) = p_T(Y|X)$), and therefore  $\frac{p_T(X)}{p_S(X)} = \frac{p_T(X, Y)}{p_S(X, Y)} $.
\subsection{Variance}
\label{sec:variance}

In this section we analyze the variance. We use $\tilde{Y}^a_i$ and $\tilde{y}^a_i$ to denote $W_i Y^a_i$ and $w_i y^a_i$ for $a \in \{0,1\}$.

\subsubsection{Finite Sample Size Variance Reduction for $d=1$}
\label{sec:d1}

In this section we will show that when $X$ is a $1$-dimensional random variable and the sample size is finite, Target Balance reduces the variance compared to complete randomization. Moreover, among all symmetric balance conditions (defined below) with rejection probability at most $\alpha$ (including Source Balance), Target Balance achieves the optimal variance reduction. The variance can be decomposed into $2$ terms (Lemma~\ref{lem:decomposition_d1}) where the second term does not depend on the balance. The first term is the variance of a 1d symmetric random variable, and Target Balance corresponds to truncating the tail, which results in the largest variance reduction (Theorem~\ref{thm:optimality_1d}). 

Let $\rho(\x, \Z) \in \{0,1\}$ denote a function that depends on only $\x$ and $\Z$, and satisfies the symmetric condition $\rho(\x, \Z) = \rho(\x, -\Z)$. This definition captures all reasonable balance conditions (including Source Balance) where $\rho =1$ denotes acceptance and $\rho = 0$ denotes rejection. Note that the constant function $\rho(\x, \Z) = 1$ for all $\x, \Z$ also satisfies the criteria  $\rho(\x, \Z) = \rho(\x, -\Z)$, and $\rho = 1$ becomes the entire sample space. We proceed to compare Target Balance with any $\rho$ satisfying the criteria above. 

First we note that by the law of total variance: 
\begin{lemma}
\label{lem:unconditional_variance}
For any function $\rho(\x, \Z) \in \{0,1\}$ satisfying  $\rho(\x, \Z) = \rho(\x, -\Z)$: 
\begin{align*}
  \var^S_{\X, \Y, \Z_{\rho}} (\hat{\tau}^T_Y ) = &\EE^S_{\X}\left[ \var^S_{\Y, \Z_{\rho}} (\hat{\tau}^T_Y | \X) \right] \\
  & +  \var^S_{\X}(\frac{1}{n}\sum_{i=1}^n W_i (\beta_1 - \beta_0)^T X_i)\;.
\end{align*}
\end{lemma}

Note that the second term does not depend on $\rho$. Therefore we focus on analyzing the variance conditioned on $\X = \x$ in this section, and the result for $\var^S_{\X, \Y, \Z_{\rho}} (\hat{\tau}^T_Y )$ easily follows from $\var^S_{\Y, \Z_{\rho}} (\hat{\tau}^T_Y | \x)$.  

Let $C_i = \frac{Y^1_i  + {Y}^0_i}{2}$,  $c_i = \frac{y^1_i  + {y}^0_i}{2}$,  $\beta = \frac{\beta_1 + \beta_0}{2}$, ${\E} = \frac{\E_1 + \E_0}{2}$ and $\sigma^2_{\E} = \var(\E)$. 
The variance of the importance weighted estimator can be written as 
\begin{lemma}
\label{lem:d1_conditional_variance}
Let $n_0 = n_1 = n/2$. For any function $\rho(\x, \Z) \in \{0,1\}$ satisfying  $\rho(\x, \Z) = \rho(\x, -\Z)$:
\begin{align*}
    &\var_{\Z_{\rho}}(\hat{\tau}^T_Y |\x, \y) =    \frac{4}{n^2}\EE_{\Z_{\rho}} \left[ \left( \sum_{i=1}^n Z_i w_i c_i  \right)^2 \bigg | \x, \y \right] 
\end{align*}
\end{lemma}
Using the law of total variance and the fact that $W_i C_i = W_i X_i \beta + W_i\E$ and $\EE[\E| \x] = 0$ we have: 
\begin{lemma}
\label{lem:decomposition_d1}
Let $n_0 = n_1 = n/2$. For any function $\rho(\x, \Z) \in \{0,1\}$ satisfying  $\rho(\x, \Z) = \rho(\x, -\Z)$:
\begin{align*}
&\var^S_{\Y, \Z_{\rho}} (\hat{\tau}^T_Y | \x) \nonumber \\
    &=  \frac{4}{n^2}\beta^2 \EE_{\Z_{\rho}} \left [ \left(\sum_{i=1}^n w_i x_i Z_i \right)^2 \bigg | \x \right ] + \frac{6}{n^2}\sigma^2_\E \sum_{i=1}^n w_i^2.
\end{align*}
\end{lemma}
We note that the design affects only the first term in the above decomposition. Let $V := \frac{2}{n} \sum_i Z_i w_i x_i =  \frac{2}{n}\tilde{\x}^T \Z$ and let $B: = V \var(V)^{-1/2}$. Recall that the Malahanobis distance $M(\frac{2}{n} (\ww\x)^T\Z) = ||B||^2$.  
Re-randomization procedure corresponds to truncating $B$ where $B$ is a mean zero random variable (as $Z_i$'s are random variables) that is symmetric about zero. 

It is easy to show that the best way to truncate a symmetric random variable $B$ to minimize the variance is to truncate the tail symmetrically $\| B\|^2 < a$ for some threshold $a$. Therefore Target Balance reduces the variance, and among all the balance conditions with rejection probability at most $\alpha$ (including Source Balance), Target Balance achieves the optimal variance reduction. 

\begin{theorem}
\label{thm:optimality_1d}
Let $n_0 = n_1 = n/2$ and $d=1$. 
   Let $\rho(\x, \Z)$ be a function satisfying $\rho(\x, \Z) = \rho(\x, -\Z)$ and $\PP(\rho = 1| \x) \ge 1 - \alpha$.  Then: 
    \begin{align*}
 \var^S_{\Y, \Z_{\phi^{\alpha}_T}} (\hat{\tau}^T_Y | \x,) &\le \var^S_{\Y, \Z_{\rho}} (\hat{\tau}^T_Y | \x).
\end{align*}
\end{theorem}
Applying Theorem~\ref{thm:optimality_1d} with $\rho$ being the constant function $\rho(\x,\Z)=1$ for all $\x,\Z,$ we have:
\begin{corollary}
\label{col:var1d}
When $d=1$ and $n_0 = n_1 = n/2$, using Target Balance reduces the variance compared to complete randomization: 
\begin{align*}
    \var^S_{\Y, \Z_{ \phi_T}} (\hat{\tau}^T_Y | \x) & \le \var^S_{\Y, \Z} (\hat{\tau}^T_Y | \x) 
    \end{align*}
\end{corollary}


\subsubsection{Asymptotic Variance Reduction for $d \ge 1$}
\label{sec:asymptotic}
In this section we show that when the sample size is large, Target Balance reduces the variance and achieves a lower variance than Source Balance. We discuss the case of finite sample size in the appendix. 

From~\citep{li2018asymptotic}, the importance weighted estimator can be decomposed into 2 components: part 1 is related to the covariates and part 2 is unrelated. Only part 1 is reduced by rerandomization while part 2 is unaffected. The covariates can be chosen to be the importance-weighted covariates (Target Balance) or the unweighted covariates (Source Balance). Since the importance-weighted covariates aligns better with the importance-weighted outcomes, part 1 will be larger and therefore the reduction by re-randomization will be larger. 

In this section we condition on $\x$ and $\y$ so the randomness only comes from $\Z$.  Similar to Section~\ref{sec:d1}, first we note that by the law of total variance: 
\begin{lemma}
\label{lem:unconditional_variance2}
For any function $\rho(\x, \Z) \in \{0,1\}$ satisfying $\rho(\x, \Z) = \rho(\x,- \Z)$:
\begin{align*}
  \var^S_{\X, \Y, \Z_{\rho}} (\hat{\tau}^T_Y) = & \EE^S_{\X, \Y} \var_{ \Z_{\rho}} (\hat{\tau}^T_Y | \X,\Y) \\ & + \var^S_{\X, \Y}(\sum_{i=1}^n W_i (Y^1_i - Y^0_i) )
\end{align*}
\end{lemma}
Since the second term does not depend on  $\rho$, we focus on analyzing the variance conditioned on $\X = \x, \Y = \y$ in this subsection. The result for $\var^S_{\X, \Y, \Z_{\rho}} (\hat{\tau}^T_Y)$ easily follows from $\var_{ \Z_{\rho}} (\hat{\tau}^T_Y | \X,\Y)$. 

Conditioning on $\x$ and $\y$, \cite{li2018asymptotic} state that if the following conditions (Condition $1$ in \citep{li2018asymptotic}) are satisfied, finite central limit theorem implies that $(\hat{\tau}^T_Y,  \frac{2}{n} \sx^T\Z)$
approaches a normal distribution as $n$ goes to infinity. Let $\text{avg}(\tilde{\y})$ and $\text{avg}(\sx)$ denote the average of the rows of $\sy$ and $\sx$. As $n \rightarrow \infty$:
\begin{itemize}[nosep,leftmargin=1em,labelwidth=*,align=left]
    \item The finite population variances and covariance $cov(\sx), cov(\sy^1), cov(\sy^0), cov(\sy^1-\sy^0), cov(\sy^1, \sx)$ and $cov(\sy^0,\sx)$ have limiting values. 
    \item $\max_{1\le i \le n} |\tilde{y}^a_i - \text{avg}(\tilde{\y})^a|^2/n \rightarrow 0$ for $a \in \{0,1\}$ and $\max_{1 \le i \le n} \norm{\tilde{x}_i - \text{avg}(\sx)}^2_2/n \rightarrow 0$ 
\end{itemize}
We apply Corollary 2 in
\cite{li2018asymptotic} to give the expression for the asymptotic variance of $\hat{\tau}^T_Y$ under Mahalanobis balance condition. Let $\text{as-var}$ denote the variance of the asymptotic sampling distribution of a sequence of random variables. 
Applying Corollary 2 in \citep{li2018asymptotic} to our case with covariates $\sx$ and $\x$ and the weighted outcome $\sy$  directly yields the following result showing both Target Balance and Source Balance reduce the variance. 
\begin{theorem}[Corollary 2 in \citep{li2018asymptotic}]
\label{thm:var_reduction}
When $n_0 = n_1 = n/2$:
\begin{align*}
& \text{as-var}_{\Z_{\phi_S}} \left( \hat{\tau}^T_Y  | \x, \y \right)\\&      = \lim_{n\rightarrow \infty} \var_{\Z}(\hat{\tau}^T_Y | \x, \y) (1 - (1-v_{d,a}) R_{\x}^2 ), \\
& \text{as-var}_{\Z_{\phi_T}}  \left( \hat{\tau}^T_Y  | \x, \y \right)\\&      = \lim_{n\rightarrow \infty} \var_{\Z}(\hat{\tau}^T_Y | \x, \y) (1 - (1-v_{d,a}) R_{\sx}^2 ),
 \end{align*}
where $R_{\sx}^2=Corr(\hat{\tau}^T_Y,  \frac{2}{n} \sx^T\Z)$, $R_{\x}^2=Corr(\hat{\tau}^T_Y,  \frac{2}{n} \x^T\Z)$ and $v_{d,a} = \frac{P(\chi^2_{d+2} \le a)}{P(\chi^2_d \le a)}$. 
\end{theorem}

We now show that Target Balance has a smaller variance than Source Balance. 
We use the following equivalent expressions for $R_{\x}^2$ and $R_{\sx}^2$. Let $Q = \frac{n}{n-1}\left(\I_d - \frac{1}{n}\1\1^T \right)$ where $\I_d$ is an identity matrix of dimension $d.$ Recall that $c_i = \frac{y^0_i + y^1_i}{2}$. Let $\cc := (c_1, \cdots, c_n)$ and $\scc = \ww \cc$. We will show that: 
\begin{lemma}
\label{lem:R2} 
When $n_0 = n_1 = n/2$:
\begin{align*}
R_{\x}^2 &= 
\sqrt{\frac{\|Q\scc \|^2 - \min_{\hbeta} \|Q\scc  - Q\x \hbeta \|^2}{\|Q\scc \|^2} }.\\
R_{\sx}^2 &= \sqrt{\frac{\|Q\scc \|^2 - \min_{\hbeta} \|Q\scc  - Q\sx \hbeta \|^2}{\|Q\scc \|^2}} .
\end{align*}
\end{lemma}
Intuitively $R^2_{\sx}$ and $R^2_{\x}$ describe how well $\scc$ is described by a linear function of $\sx$ and $\x$, respectively. Because of our model, a linear model in terms of $\tilde{\x} = \ww\x$ fits $\scc = \ww\cc$ better than a linear model in terms of $\x$. Therefore, $R^2_{\tilde{\x}}$ will be larger than $R^2_{\x}$ and using $\phi_T$ will result in a smaller variance than $\phi_S$. 


Therefore with the same rejection probability $\alpha$, Target Balance has a lower variance than Source Balance. 

\begin{theorem}
\label{thm:asymptotic_optimal}
When $n_0 = n_1 = n/2$, if $X_i$, $Y_i$ and $W_i$ have finite eighth moment according to the source distribution, with the same rejection probability $\alpha$:
\begin{align*}
    \text{as-var}_{\Z_{\phi^{\alpha}_T}}  \left( \hat{\tau}^T_Y |\X, \Y \right) \le \text{as-var}_{\Z_{\phi^{\alpha}_S}}  \left(\hat{\tau}^T_Y  | \X, \Y \right)
\end{align*}
almost surely. 
\end{theorem}

\section{Simulations}
\label{sec:simulation}
\begin{figure*}
    \centering
    \includegraphics[width=0.9\textwidth]{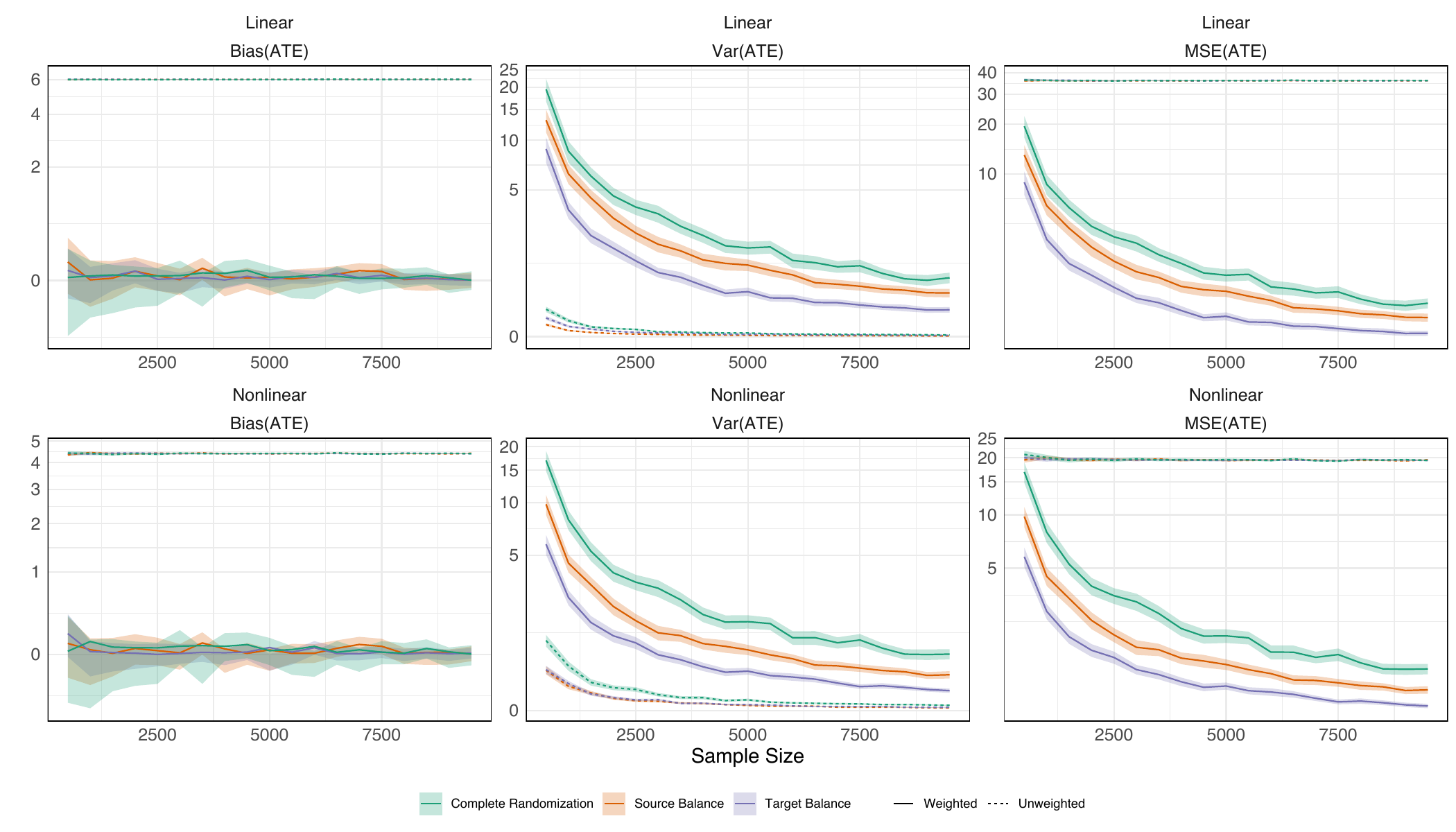}
    \caption{Bias, Variance and MSE as a function of the sample size.  All unweighted estimators are biases because they measure the ATE of the source distribution. As there is no importance weight threshold, all weighted estimators are unbiased (Theorem~\ref{thm:unbiased}) but the weighted estimator with Target Balance has the lowest variance. \emph{The $y$ axes are in log scale.}}
    \label{fig:sample_size}
\end{figure*}

\begin{figure*}
    \centering
    \includegraphics[width=0.9\textwidth]{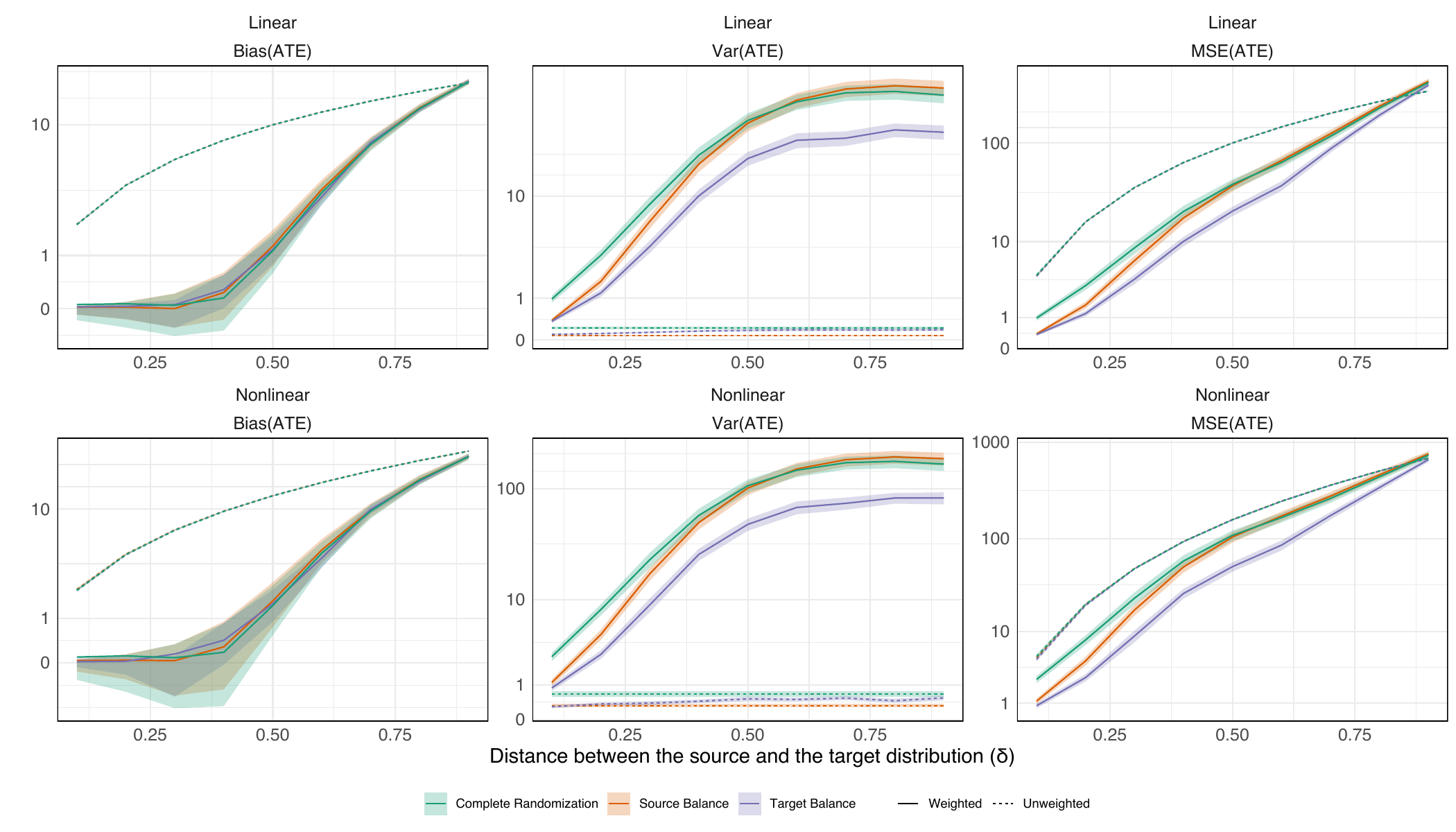}
    \caption{Bias, Variance and MSE as a function of the distance $\delta$ (defined in Section~\ref{sec:simulation}) between the source and the target distribution. Because of the importance weight threshold, the biases of the importance weighted methods increase as the $\delta$ increase. If the distance is too large, the bias of the importance weighted estimators is large, leading to high MSE. However when the distance is not too large, the weighted estimator with Target Balance has the lowest MSE. \emph{The $y$ axes are in log scale. }}
    \label{fig:distance}
\end{figure*}

\begin{figure*}
    \centering
    \includegraphics[width=0.8\textwidth]{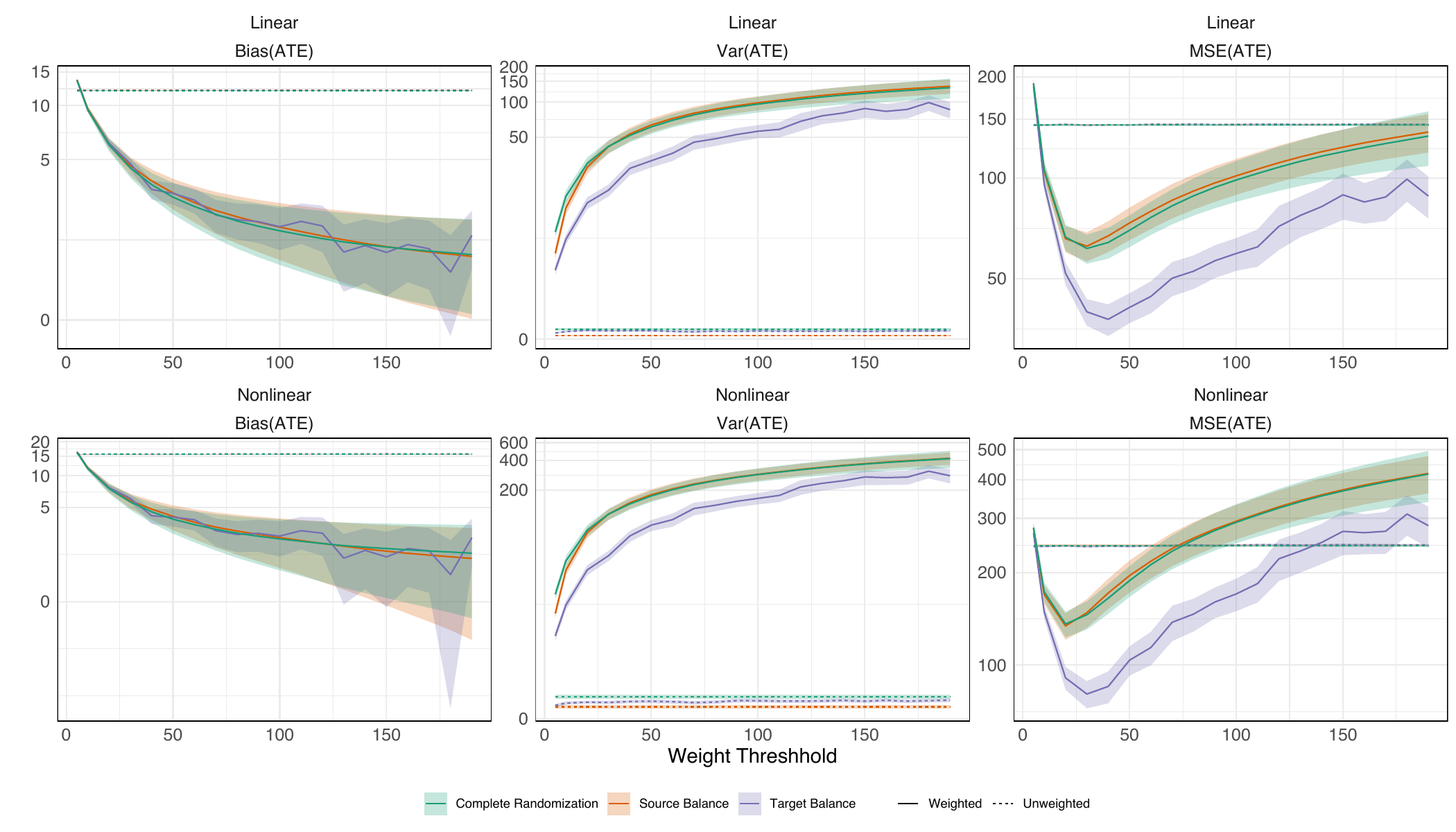}
    \caption{Bias, Variance and MSE as a function of importance weight threshold. As the threshold increases, the bias of the weighted methods decreases and the variance of the weighted methods increases.  Therefore there is a threshold when the MSE is minimized. The weighted estimator with Target Balance has the lowest MSE for a reasonably good threshold. \emph{The $y$ axes are in log scale. }}
    \label{fig:threshold}
\end{figure*}
We perform simulations on the two following models:

\textbf{Linear Model}
    \begin{align*}
    Y^0 = X + Norm(0,1) ; Y^1 = 3X + Norm(0,1)
    \end{align*}
\textbf{Nonlinear Model}
    \begin{align*}
    &Y^0 = X^TX + Norm(0,1) ; Y^1 = 2X^TX + Norm(0,1)
\end{align*}
We use the following source and target distributions for $X$. In the source distribution, $X \sim \text{MultivariateNorm}(\mathbf{1}, \I)$ where $\I$ is the identity matrix. In the target distribution, $X \sim \text{MultivariateNorm}(\mathbf{1} + \delta, \I)$ where $\delta$ is a parameter that will be specified later.

We randomly choose an assignment such that $n_1= n_0=n/2$. To select the random assignment with the top balance, instead of choosing a fixed threshold $\alpha$, we select the rejection probability $\alpha=0.99$ as in Def.~\ref{def:target_balance}. To implement this, we draw $100/(1-\alpha)$ assignments at random, calculate their Mahalanobis distance and pick one among the smallest $100$ uniformly at random. 

If the source and the target distributions are far away, importance weighting can induce large variance. We use the weight clipping technique, in which if the importance weight is larger than a threshold, it will be set to that threshold. It will induce bias but reduce variance, and therefore reduce mean square error (MSE). 

We compare 6 methods (WE, CR), (WE, SB), (WE, TB), (UE, CR), (UE, SB) and (UE, TB) by combining the following 2 properties:

\textbf{Weighted and Unweighted}. 
    \begin{itemize}[nosep,leftmargin=1em,labelwidth=*,align=left]  
        \item \emph{Weighted Estimator (WE).} We consider the importance weighted estimator in Eq.~\ref{eq:importance_sample}.
        \item \emph{Unweighted Estimator (UE).} We consider the unweighted estimator which is equivalent to Eq.\ref{eq:importance_sample} with all weights set to one.
\end{itemize}

 \textbf{Complete Randomization, Source Balance and Target Balance}. 
    \begin{itemize}[nosep,leftmargin=1em,labelwidth=*,align=left]  
        \item \emph{Complete Randomization (CR).} This is the randomized assignment without balancing.
         \item \emph{Source Balance (SB).} This is the rerandomization algorithm seeking Source Balance. 
    
    \item \emph{Target Balance (TB).} This is the rerandomization algorithm seeking Target Balance as in Definition~\ref{def:target_balance}. 
    \end{itemize}

We study the MSE of our methods in relation to the 3 following parameters: the sample size $n$, the importance weights threshold and the distance $\delta$.  Recall that in the source distribution, $X \sim \text{MultivariateNorm}(\mathbf{1}, \I)$ where $\I$ is the identity matrix and in the target distribution, $X \sim \text{MultivariateNorm}(\mathbf{1} + \delta, \I)$.

    \textbf{Sample Size.} In this experiment for both models we vary the sample size from $500$ to $9500$ with step size $500$ and set the number of covariates to $10$. For the linear model, $\delta = 0.3$. For the nonlinear model, $\delta = 0.2$. $\delta$ is chosen to be small enough so that we do not need weight clipping. For each sample size we repeat the experiment $500$ times. There is no importance weight threshold. The results are shown and discussed in Figure~\ref{fig:sample_size}.

    \textbf{Threshold.} In this experiment for both models we vary the importance weight threshold from $5$, then $10$ to $190$ with step size $10$. We set the number of covariates to $10$ and the sample size to be $1000$ and $\delta = 0.6$. $\delta$ is chosen to be large enough so that weight clipping is necessary. For each threshold we repeat the experiment $500$ times. The results are shown and discussed in Figure~\ref{fig:threshold}. 

    \textbf{Distance $\delta$.} In this experiment for both models we vary $\delta$ from $0.1$ to $0.9$ with step size $0.1$. We set the number of covariates to $10$, the sample size to be $1000$ and the importance weight threshold to be $40$. From the weight threshold experiment, we know that if the weight threshold is too large, the variance is too high while if the weight threshold is too small, the bias will be too high. Therefore we pick the value $40$ as a reasonable weight threshold. For each threshold we repeat the experiment $500$ times. The results are shown and discussed in Figure~\ref{fig:distance}. 
Across all simulations, Target Balance with the Weighted Estimator substantially reduces the MSE. 


\section{Conclusion}
\label{sec:conclusion}

In this work, we've shown that a desire for generalizability should change the way experiments are designed and run.
In particular, we argue that balance should be sought on the target population rather than the samples in which randomization will actually be performed.
We present a method for designing an experiment along these lines, show theoretically that it is unbiased and more efficient than sample balancing.
\bibliographystyle{plainnat}  
\bibliography{references}
\clearpage
\appendix
\onecolumn


\begin{center}
{\Large Supplement to "Designing Transportable Experiments Under S-Admissability"}
\end{center}
In Section~\ref{apx:additional_results} we discuss the variance reduction for $d \ge 1$ when the sample size is finite. In Section~\ref{apx:expectation} we show the proofs of Section~\ref{sec:expectation}. In Section~\ref{apx:d1} we show the proofs of Section~\ref{sec:d1}. In Section~\ref{apx:asymptotic} we show the proofs of Section~\ref{sec:asymptotic}. In Section~\ref{apx:additional_results_proofs} we show the proofs of Appendix~\ref{apx:additional_results}.

For a random variable $R$ with value $r$, we write the expectation, variance and covariance conditioning on $r$ as a short-hand for conditioning on $R=r$. On the other hand, the expectation, variance and covariance conditioning on $R$ are functions of $R$ and therefore are random variables. For example,$\EE[\htau^T_Y | \X, \Y]$ is a function of $\X$ and $\Y$, $\EE[\htau^T_Y | \X, \y] = \EE[\htau^T_Y | \X, \Y= \y] $ is a function of $\X$, while $ \EE[\htau^T_Y | \x,  \y] =\EE[\htau^T_Y | \X = \x,\Y=  \y]$ is a value. 

Conditioning on $\x$ and $\y$, the randomness only comes from $\Z$. Therefore  $\var_{\Z_{\rho}} (. | \x, \y), Cov_{\Z_{\rho}} (. | \x, \y)$ and $\EE_{\Z_{\rho}} (. | \x, \y)$ can be written as $\var_{\Z} (. | \x, \y, \rho = 1), Cov_{\Z} (. | \x, \y, \rho = 1)$ and $\EE_{\Z} (. | \x, \y, \rho = 1)$ respectively. We use both notations in the proofs. 

For a random variable $R$, we use $Cov(R)^{-1/2}$ to denote the Cholesky square root of $Cov(R)^{-1}$. 

We restate the model and some notations here for convenience. Let the model be: 
  \begin{align*}
        Y^1_i = X_i^T \beta_1 + \E^1_i & \quad \quad
        Y^0_i =  X_i^T \beta_0 + \E^0_i
        \end{align*}
Let $\epsilon^1_i$ and $\epsilon^0_i$ be the values taken by random variables $\E^1_i$ and $\E^0_i$. 
Let $C_i = \frac{Y^0_i + Y^1_i}{2}$, $\TC_i = W_i C_i$, $\CC := (C_1, \cdots, C_n)$ and $\SC = (\TC_1, \cdots, \TC_n)$. Let $c_i, \tilde{c}_i , \cc$ and $\scc$ be the values taken by $C_i,  \TC_i, \CC$ and $\SC$. Then
\begin{align*}
    C_i &= X_i^T \beta + \E_i  & \quad \quad c_i = x_i^T \beta + \epsilon_i\\
    \TC_i &= \tilde{X}_i^T \beta + \SE_i & \quad \quad \tilde{c}_i = \tilde{x}_i^T \beta + \tilde{\epsilon}_i
\end{align*}
where $\beta = \frac{\beta_1 + \beta_0}{2}, \E_i= \frac{\E^1_i + \E^0_i}{2}, \tilde{X}_i = W_i X_i$ and $\SE_i = W_i \E_i$. Let $\epsilon_i$  and $\tilde{\epsilon}_i = w_i \epsilon_i$ be the value taken by $\E_i$ and $\SE_i$. Let $\SBE = (\SE_1, \cdots, \SE_n)$. 
\section{Additional Results: Finite Sample Size Variance Reduction for $d \ge 1$}
\label{apx:additional_results}
In this section we discuss the finite sample case when $X$ is a multivariate random variable, which is a generalization of the result in Section~\ref{sec:d1} when $d=1$. We show that when the sample size is finite, if $\beta$ points to all directions with equal probability, then a balance condition which also consider the target population and is similar to Target Balance achieves the optimal variance reduction in expectation over $\beta$. The proofs are in Appendix~\ref{apx:additional_results_proofs}. 

We will use the variance decomposition in the matrix form similar to \citep{harshaw2019balancing} and provide intuition about the effect of balancing on the variance.
The following lemma is the general case when $d \ge 1$ of Lemma~\ref{lem:decomposition_d1} in Section~\ref{sec:d1}. 

\begin{lemma}
\label{lem:variance_decomposition}
For any function $\rho(\x, \Z) \in \{0,1\}$ satisfying  $\rho(\x, \Z) = \rho(\x, -\Z)$:
    \begin{align*}
    \var^S_{\Y, \Z_{\rho}} (\hat{\tau}^T_Y | \x) =   \beta^T Cov_{\Z_{\rho}} ( V |\x) \beta  + \frac{6}{n^2}\sigma_{\E}^2 \sum_{i=1}^n w_i^2, 
\end{align*}
for $V := \frac{2}{n}(\ww \cdot \x)^T\Z  = \frac{2}{n}\tilde{\x}^T\Z.$
\end{lemma}

Since the design affects only the first term in the above expression, we focus on the the random variable $V$. $V$ is now a $d$-dimensional vector and  $\beta$ is unknown. 

To understand the first term, we use the same decomposition of $\beta^T Cov_{\Z_{\rho}}(V|\x) \beta$ as in~\citep{harshaw2019balancing}.  Let $\e_1, ..., \e_n$ and $\lambda_1, .., \lambda_n$ be the normalized eigenvectors and corresponding eigenvalues of  matrix $Cov_{\Z_{\rho}} (V | \x) $. Since  $Cov_{\Z_{\rho}} (V | \x) $ is symmetric, the eigenvectors form an orthonormal basis so we can write $\beta$ as a linear combination of $\e_1, .., \e_n$ and get:
\begin{align*}
   \beta = \norm{\beta}\sum_{i=1}^n \eta_i \e_i 
\end{align*}
where $\eta_i = \langle \beta, \e_i \rangle/\norm{\beta}$ is the coefficient that captures the alignment of the weighted outcome $\beta$ with respect to the eigenvector $\e_i$. Therefore:
\begin{align*}
   \beta^T  Cov_{\Z_{\rho}} (V | \x)  \beta = \norm{\beta}^2\sum_{i=1}^n \eta_i^2 \lambda_i
\end{align*}

In the worst case, $\beta$ can align with the eigenvector of   $Cov_{\Z_{\rho}} (V | \x)$ with the largest eigenvalue. Therefore a good design is one with $\rho$ that minimize the largest eigenvalue of $Cov_{\Z_{\rho}} (V | \x)$. We leave this for future works. In this work we  consider the average case direction - when $\beta$ with norm $\|\beta\| = l $ can point in any direction with equal probability. In that case, we have  
\begin{lemma}
\label{lem:expected_beta}
\begin{align}
\EE_{\|\beta\|= l}  \beta^T Cov_{\Z_{\rho}} (V | \x)  \beta  = \frac{l^2}{2} \text{Trace}(Cov_{\Z_{\rho}} (V | \x) ).
\end{align}
\end{lemma}
We can then ask for the balance event $\Omega$ which results in minimizing the trace of  $Cov_{\Z}(V |\x, \Omega)$, which is shown in the following lemma. Note that when $d=1$, the trace of  $Cov_{\Z}(V |\x, \Omega)$ is the variance $\var_{\Z}(V |\x, \Omega)$, and this result is the general case of minimizing the variance of a $1$-dimensional random variable in Section~\ref{sec:d1}.
\begin{lemma}
\label{lem:trace_optimality}
Let $U \in \mathcal{R}^d$ be a random variable such that $\EE[U] = 0$. 
Let $u_{\alpha}$ be such that $\PP(\| U \|^2 < u_{\alpha}) = 1 - \alpha$.
     
     Let $\Omega$ be an event such that $\PP(\Omega ) \ge 1- \alpha$ and $\EE[U | \Omega] = 0$. Then:
\begin{align*}
     Trace( Cov(U |\| U \|^2 < u_{\alpha}) \le Trace( Cov(U | \Omega) )\end{align*}
\end{lemma}

It follows from Lemma~\ref{lem:variance_decomposition}, Lemma~\ref{lem:expected_beta} and Lemma~\ref{lem:trace_optimality} that we can minimize $\EE_{\beta} \var^S_{\Y, \Z} (\hat{\tau}^T_Y | \x, \Omega) $ by defining the following balance condition:
\begin{definition}[Alternate Target Balance]
With a rejection threshold  $\alpha$, define the balance condition 
$$
\phi'^{\alpha}_T =
    \begin{cases}
      1, & \text{if } \| V \|^2 < a \\
      0, & \text{otherwise}
    \end{cases}
$$
where $a$ be such that $\PP(\phi'^{\alpha}_T=1 | \x) = 1 - \alpha$.
\end{definition}
Recall that Target Balance use the condition $\|B\|^2 < a$ where $B = V Cov_{\Z}(V)^{-1/2}$ is the normalized random variable of $V$. Note since that $V = \frac{2}{n}\tilde{\x}^T\Z$, Alternate Target Balance also considers the target population in the design phase. However Alternate Target Balance is not invariant under linear transformations of the covariates $x_i$'s while Target Balance is. 

We have the following Theorem which is a generalization of Theorem~\ref{thm:optimality_1d} in Section~\ref{sec:d1}. 
\begin{theorem}
\label{thm:expected_beta_optimal}
 Let $\|\beta\| = l$ and $\beta$ points in any direction with equal probability and $n_0 = n_1 = n/2$. 

Let $\rho(\X, \Z)$ be a function satisfying $\rho(\X, \Z) = \rho(\X, -\Z)$ and $\PP(\rho = 1| \x) \ge 1 - \alpha$. Then
\begin{align*}
     \EE_{\beta} \var^S_{\Y, \Z_{\phi'^{\alpha}_T}} (\hat{\tau}^T_Y | \x) &\le \EE_{\beta} \var^S_{\Y, \Z_{\rho}} (\hat{\tau}^T_Y | \x)
\end{align*}
\end{theorem}

Similar to Section~\ref{sec:d1}, applying Theorem~\ref{thm:optimality_1d} with $\rho$ being the constant function $\rho(\x,\Z)=1$ for all $\x,\Z,$ we have:
\begin{corollary}
\label{col:var_d_greater_1}
Let $\|\beta\| = l$ and $\beta$ points in any direction with equal probability. When  $n_0 = n_1 = n/2$, using Alternate Target Balance reduces the variance compared to complete randomization in expectation over $\beta$. 
\begin{align*}
     \EE_{\beta} \var^S_{ \Z_{\phi'_T}, \Y} (\hat{\tau}^T_Y | \x) &\le \EE_{\beta} \var^S_{ \Z, \Y} (\hat{\tau}^T_Y | \x)
\end{align*}
\end{corollary}

Recall that the first term in the decomposition in Lemma~\ref{lem:variance_decomposition}  is equal to: 
\begin{align*}
     &\beta^T Cov_{\Z_{\rho}}(V|\x) \beta    = \gamma^T Cov_{\Z_{\rho}} (B| \x) \gamma = \gamma^T Cov_{\Z}(B | \x, \rho = 1) \gamma
\end{align*}
where $\gamma =  \beta^T Cov_{\Z}(V)^{1/2}$ and $B = VCov_{\Z}(V)^{-1/2}$. 

When the sample size is large, $B$ converges to a standard normal distribution. Recall that Target Balance is equal to truncating $\|B\|^2 < a$. So $ Cov_{\Z_{\phi_T}}(B | \x) $ is the covariance of a standard normal random variable $B$ truncated by $\|B\|^2 < a$. From Theorem 3.1 in \citep{morgan2012rerandomization} when $B$ is a standard normal distribution, $ Cov(B | \x, \phi_T = 1) = v Cov(B| \x)$ for some $v < 1$, so the variance is reduced. However we do not need to go through this analysis because \citep{li2018asymptotic} already has variance reduction results for the case when the sample size is large. In Section~\ref{sec:asymptotic} we use the result from \citep{li2018asymptotic} directly to show that Target Balance achieves a smaller variance than Source Balance. 

\section{Proofs of Section~\ref{sec:expectation}}
\label{apx:expectation}
In this section we prove Theorem~\ref{thm:unbiased}. We made use of the following lemma from~\cite{morgan2012rerandomization}:
 \begin{lemma}[from the proof of Theorem 2.1 in \cite{morgan2012rerandomization}]
\label{lem:helper}
Let  $\A := (A_1, ..., A_n)^T \in \mathcal{R}^n$. Let $n_1 = n_0 = n/2$. For any function $\rho(\x, \A) \in \{0,1\}$ satisfying $\rho(\x, \A) = \rho(\x,1- \A)$:
\begin{align*}
    \EE^S_{\A} [ A_i | \x, \y, \rho = 1] = \frac{1}{2} 
\end{align*}
\end{lemma}
We also prove the following lemma in order to prove Theorem~\ref{thm:unbiased}:
\begin{lemma}
\label{lem:expectation}
For any function $\rho(\x, \A) \in \{0,1\}$ satisfying $\rho(\x, \A) = \rho(\x,1- \A)$:
\begin{align*}
\mathbb{E}_{\A| \rho = 1}[\hat{\tau}^T_Y | \X, \Y ]
&= \frac{1}{n} \sum_{i=1}^n W_i  (Y^1_i  - Y^0_i) \\
 \EE^S_{\Y, \A| \rho =1}[\hat{\tau}^T_Y | \X]
&= \frac{1}{n} \sum_{i=1}^n W_i  (\beta_1- \beta_0)^T X_i 
\end{align*}
\end{lemma}
\begin{proof}
 From Lemma~\ref{lem:helper}, $\EE [ A_i | \X, \Y, \rho = 1] = \EE [ A_i | \X, \rho = 1]= \frac{1}{2}$. Therefore:
\begin{align*}
\mathbb{E}_{\A| \rho = 1}[\hat{\tau}^T_Y | \X, \Y ]
&=  \frac{1}{n_1} \sum_{i=1}^n \mathbb{E}_{\A} \left [W_i A_i Y^1_i \bigg | \X, \Y, \rho = 1 \right] - \frac{1}{n_0} \sum_{i=1}^n \mathbb{E}_{\A} \left[W_i  (1-A_i) Y^0_i \bigg | \X, \Y, \rho = 1 \right] \\
&=  \frac{1}{n_1} \sum_{i=1}^n W_i  Y^1_i  \EE_{\A}  \left [A_i \big | \X, \Y, \rho = 1 \right]- \frac{1}{n_0} \sum_{i=1}^n W_i  Y^0_i  \EE_{\A} \left [1- A_i \big | \X, \Y, \rho = 1 \right]\\
&=  \frac{1}{n} \sum_{i=1}^n W_i  (Y^1_i  - Y^0_i) \\
   \mathbb{E}^S_{\A| \rho = 1, \Y}[\hat{\tau}^T_Y | \X ]  &=  \mathbb{E}^S_{\Y} \left[ \mathbb{E}_{\A}[\hat{\tau}^T_Y | \X, \Y,  \rho = 1 ] | \X\right ] \\
    &= \EE^S_{\Y} \left[ \frac{1}{n} \sum_{i=1}^n W_i  (Y^1_i  - Y^0_i) | \X\right ] \\
    &=  \frac{1}{n} \sum_{i=1}^n W_i  (\beta_1- \beta_0)^T X_i 
\end{align*}
\end{proof}
\begin{proof}[Proof of Theorem~\ref{thm:unbiased}]
 Let $D_S$ and $D_T$ be the supports of the source and target distributions. Since $p_T(X) > 0 \rightarrow p_S(X) > 0 $ and $p_T(Y|X) = p_S(Y|X)$, we have $D_T \subseteq D_S$. Using Lemma~\ref{lem:expectation}:
\begin{align*}
    \mathbb{E}^S_{\X, \Y, \Z_{\phi_T} } \left[\hat{\tau}^T_Y  \right] 
   = & 
    \mathbb{E}^S_{\X, \Y} \mathbb{E}_{\A_{\phi_T} }[\hat{\tau}^T_Y | \X, \Y ] \\
     = & \frac{1}{n} \sum_{i=1}^n \mathbb{E}^S_{\X, \Y} \left [  W_i  (Y^1_i  - Y^0_i) \right] \\
    = & \frac{1}{n} \sum_{i=1}^n \int_{(x,y) \in D_S} \left(  \frac{p_T(x)}{p_S(x)} (y^1  - y^0) \right) p_S(x,y) dxy \\
    = & \frac{1}{n} \sum_{i=1}^n \int_{(x,y) \in D_S} \left(  \frac{p_T(y|x) p_T(x)}{p_S(y|x)p_S(x)} (y^1  - y^0) \right) p_S(x,y) dxy \text{ because } p_T(y|x) = p_S(y|x) \\
    = & \frac{1}{n} \sum_{i=1}^n \int_{(x,y) \in D_S} \left(  \frac{p_T(y, x)}{p_S(y,x)} (y^1  - y^0) \right) p_S(x,y) dxy \\
    = & \frac{1}{n} \sum_{i=1}^n \int_{(x,y) \in D_S} {p_T(x,y)}(y^1  - y^0)   dxy \\
     = & \frac{1}{n} \sum_{i=1}^n \int_{(x,y) \in D_T} {p_T(x,y)}(y^1  - y^0)   dxy \text{ because $D_T \subseteq D_S$}\\
    = & \tau^T_Y
\end{align*}
\end{proof}

\section{Proofs of Section~\ref{sec:d1}}
\label{apx:d1}
In this section we prove Lemma~\ref{lem:unconditional_variance}, Lemma~\ref{lem:d1_conditional_variance}, Lemma~\ref{lem:decomposition_d1}, Theorem~\ref{thm:optimality_1d} and Corollary~\ref{col:var1d}. Note that the results in this section are the special case when $d=1$ of the results in Section~\ref{apx:additional_results}.  Lemma~\ref{lem:d1_conditional_variance} is a special case when $d=1$ of Lemma ~\ref{lem:variance_decomposition_condition_on_y}. Lemma~\ref{lem:decomposition_d1} is a special case of Lemma~\ref{lem:variance_decomposition} and Theorem~\ref{thm:optimality_1d} is a special case of Theorem~\ref{thm:expected_beta_optimal}. However in this section we state the full proofs for the case $d=1$ so that the readers do not need to read the proofs of Section~\ref{apx:additional_results} in order to understand Section~\ref{sec:d1} in the main paper. 

\begin{proof}[Proof of Lemma~\ref{lem:unconditional_variance}]
By law of total variance:
\begin{align*}
   \var^S_{\Z_{\rho}, \X, \Y} (\hat{\tau}^T_Y ) = \EE^S_{\X} \var^S_{\Y, \Z_{\rho}} (\hat{\tau}^T_Y | \X)  + \var^S_{\X} \left (\EE^S_{\Y, \Z_{\rho} }[\hat{\tau}^T_Y | \X] \right)
\end{align*}
Since $\rho(\x, \Z) = \rho(\x, -\Z)$, from Lemma~\ref{lem:expectation}: 
\begin{align*}
 \EE^S_{\Y, \Z_{\rho}}[\hat{\tau}^T_Y | \X]
= \frac{1}{n} \sum_{i=1}^n W_i  (\beta_1- \beta_0)^T X_i 
\end{align*}
Therefore:
\begin{align*}
     \var^S_{\X} \left (\EE^S_{\Y, \Z_{\rho} }[\hat{\tau}^T_Y | \X]\right) = \var^S_{\X} \left(\frac{1}{n} \sum_{i=1}^n W_i  (\beta_1- \beta_0)^T X_i)   \right) 
\end{align*}
\end{proof}
\begin{proof}[Proof of Lemma~\ref{lem:d1_conditional_variance}]
By definition:
\begin{align*}
\var_{\Z}(\hat{\tau}^T_Y |  \x, \y, \rho=1) 
&= \EE_{\Z} \left[(\hat{\tau}^T_Y - \mathbb{E}_{\Z}[\hat{\tau}^T_Y | \x, \y, \rho = 1 ] )^2 | \x, \y, \rho = 1 \right] 
\end{align*}
From Lemma~\ref{lem:expectation}
\begin{align*}
\EE_{\Z} [\hat{\tau}^T_Y | \x , \y, \rho = 1] = \frac{1}{n} \left(\sum_{i=1}^n  w_i y^1_i - \sum_{i=1}^n w_i y^0_i \right) 
\end{align*}
On the other hand conditioning on $\X = \x$ and $\Y = \y$ and let $y^*_i$ denote the observed outcome of sample $i$: 
\begin{align*}
    \hat{\tau}^T_Y &= \frac{2}{n}(\sum_{Z_i=1} w_i y^*_i - \sum_{Z_i = -1} w_i y^*_i) \\
    &= \frac{2}{n} \sum_{i=1}^n w_i A_i y^1_i - \frac{2}{n} \sum_{i=1}^n w_i (1-A_i) y^0_i
\end{align*}
Therefore: 
\begin{align*}
\var_{\Z}(\hat{\tau}^T_Y |  \x, \y, \rho = 1)  &= \EE_{\Z} \left[ \left( \frac{2}{n}(\sum_{i=1}^n w_i A_i y^1_i -  \sum_{i=1}^n w_i (1-A_i) y^0_i) - \frac{1}{n} \sum_{i=1}^n w_i (y^1_i - y^0_i) \right)^2 \bigg| \x, \y, \rho = 1 \right] \\
&= \EE_{\Z} \left[ \left( \frac{1}{n}(\sum_{i=1}^n w_i (2A_i - 1) y^1_i + \frac{1}{n} \sum_{i=1}^n w_i (2A_i-1) y^0_i) \right)^2 \bigg| \x, \y, \rho = 1 \right] \\
&= \frac{4}{n^2} \EE_{\Z} \left[ \left ( \sum_{i=1}^n w_i Z_i \frac{y^1_i + y^0_i}{2}\right)^2 \bigg | \x, \y, \rho = 1\right] \\
&= \frac{4}{n^2} \EE_{\Z} \left[ \left ( \sum_{i=1}^n  Z_i w_i c_i\right)^2 \bigg | \x, \y, \rho = 1 \right]
\end{align*}
where ${c}_i = \frac{y^1_i + y^0_i}{2}$. 
\end{proof}

\begin{proof}[Proof of Lemma~\ref{lem:decomposition_d1}]
By law of total variance:
\begin{align*}
\var^S_{\Y, \Z_{\rho}} (\hat{\tau}^T_Y | \x )  
&= \EE^S_{\Y} \left[ \var_{\Z_{\rho}} (\hat{\tau}^T_Y | \x, \Y)  | \x \right] + \var^S_{\Y} ( \EE_{\Z_{\rho}} [\hat{\tau}^T_Y | \x, \Y] |\x) \\
&=  \EE^S_{\Y} \left[ \var_{\Z_{\rho}} (\hat{\tau}^T_Y | \x, \Y)  | \x \right]  + \var^S_{\Y} \left( \frac{1}{n} \sum_{i=1}^n w_i (Y^1_i  - Y^0_i) | \x  \right) \\
&=  \EE^S_{\Y} \left[ \var_{\Z_{\rho}} (\hat{\tau}^T_Y | \x, \Y)  | \x \right]  + \frac{1}{n^2} \sum_{i=1}^n w_i^2 \var (\E^{1}_i  - \E^{0}_i)  \\
&=  \EE^S_{\Y} \left[ \var_{\Z_{\rho}} (\hat{\tau}^T_Y | \x, \Y)  | \x \right]  + \frac{2}{n^2}\sigma_{\E}^2 \sum_{i=1}^n w_i^2 
\end{align*}

Recall that $\tilde{C}_i =  \beta \tilde{X}_i + \SE_i$. From Lemma~\ref{lem:d1_conditional_variance}: 
\begin{align}
\label{eq:decomposition1}
    &\var_{\Z}(\hat{\tau}^T_Y |\x, \Y, \rho =1) \notag \\
    &=   \frac{4}{n^2}\EE_{\Z} \left[ \left( \sum_{i=1}^n Z_i \tilde{C}_i  \right)^2 \bigg | \x, \Y, \rho = 1 \right] \notag \\
    &=   \frac{4}{n^2}\EE_{\Z} \left[ \left(  \Z^T \SC  \right)^2 \bigg | \x, \Y, \rho = 1 \right] \notag \\
    &= \frac{4}{n^2}\EE_{\Z} \left[ \left( \Z^T \beta \sx  + \Z^T  \SBE \right)^2 \bigg | \x, \Y, \rho = 1 \right] \notag\\
    &= \frac{4}{n^2}\beta^2\EE_{\Z} \left[ \left( \Z^T\sx \right)^2 \bigg | \x,  \rho = 1 \right] + \frac{4}{n^2}\EE_{\Z} \left[ \left( \Z^T\SBE \right)^2 \bigg | \x, \Y, \rho = 1 \right] + \frac{4}{n^2}2 \EE_{\Z} \left[ \sx^T\Z\Z^T \SBE \bigg | \x, \Y, \rho = 1 \right]
\end{align}
Now we consider $\EE^S_{\Y} \left[ \var_{\Z} (\hat{\tau}^T_Y | \x, \Y, \rho = 1)   | \x \right] $. The third term in Eq.~\ref{eq:decomposition1} becomes: 
\begin{align*}
  \frac{4}{n^2} 2\EE^S_{\Y} \left[ \EE_{\Z} \left[ \sx^T\Z\Z^T \SBE \bigg | \x, \Y, \rho = 1 \right] \bigg | \x \right] &= \frac{8}{n^2}  \EE_{\Z} \left[ \sx^T\Z\Z^T \bigg | \x,  \rho = 1 \right] \EE^S_{\Y}[\SBE | \x] \\
  &= 0\text{ because } \EE^S_{\Y}[\SBE | \x] = \mathbf{0}
\end{align*}

The second term in Eq.~\ref{eq:decomposition1} becomes:
\begin{align*}
&\frac{4}{n^2} \EE^S_{\Y} \left[ \EE_{\Z} \left[ \left( \Z^T \SBE \right)^2 \bigg | \x, \Y, \rho = 1 \right] \bigg | \x \right ] \\
  &=\frac{4}{n^2} \EE^S_{\Y} \left [\EE_{\Z} \left[ \left( \sum_{i=1}^n Z_i w_i \E_i \right)^2 \big | \x, \Y, \rho = 1 \right] \bigg | \x \right] \\
  &=\frac{4}{n^2}\EE^S_{\Y} \left [\EE_{\Z} \left[\sum_{i=1}^n (Z_i w_i \E_i)^2 \big | \x, \Y, \rho = 1 \right] \bigg | \x \right] + \frac{4}{n^2}\EE^S_{\Y} \left [\EE_{\Z} \left[\sum_{i \ne j } (Z_i w_i \E_i)  (Z_j w_j \E_j) \big | \x, \Y, \rho = 1 \right] \bigg | \x \right] \\
    &=\frac{4}{n^2}\EE^S_{\Y} \left [\EE_{\Z} \left[\sum_{i=1}^n (Z_i w_i \E_i)^2 \big | \x, \Y, \rho = 1 \right] \bigg | \x \right] +  \frac{4}{n^2}\sum_{i \ne j } \EE_{\Z}[Z_i Z_j | \x,  \rho = 1 ] w_i  w_j   \EE^S_{\Y} \left [\E_i  \E_j | \x \right]\\
  &=\frac{4}{n^2}\EE^S_{\Y} \left [\EE_{\Z} \left[\sum_{i=1}^n (Z_i w_i \E_i)^2 \big | \x, \Y, \rho = 1 \right] \bigg | \x\right] + 0 \text{                  because } \EE^S_{\Y}[\E_i\E_j | \x ] = \EE^S_{\Y}[\E_i | \x ]\EE^S_{\Y}[\E_j | \x ]=0 \\
  &= \frac{4}{n^2}\EE^S_{\Y} \left[\sum_{i=1}^n (w_i \E_i)^2 ] \big | \x \right] \text{ because $Z_i^2 = 1$} \\
  &= \frac{4}{n^2} \sigma_{\E}^2 \sum_{i=1}^n w_i^2
\end{align*}
The first term in Eq.~\ref{eq:decomposition1} becomes: 
\begin{align*}
    \frac{4}{n^2} \EE^S_{\Y} \left[\beta^2\EE_{\Z} \left[ \left( \Z^T\sx \right)^2 \bigg | \x,  \rho = 1 \right]\bigg| \x \right]  &= \frac{4}{n^2}\beta^2\EE_{\Z} \left[ \left( \Z^T\sx \right)^2 \bigg | \x,  \rho = 1 \right] \\
    &=\frac{4}{n^2}\beta^2 \EE_{\Z} \left[ ( \sum_{i=1}^n Z_i w_i  x_i )^2 \big | \x, \rho = 1 \right] 
\end{align*}
Putting all $3$ terms together: 
\begin{align*}
   \EE^S_{\Y} \left[ \var_{\Z} (\hat{\tau}^T_Y | \x, \Y,\rho = 1)  | \x \right]  = \frac{4}{n^2}\beta^2 \EE_{\Z} \left[ ( \sum_{i=1}^n Z_i w_i  x_i )^2 \big | \x,  \rho = 1 \right] + \frac{4}{n^2} \sigma_{\E}^2 \sum_{i=1}^n w_i^2
\end{align*}
Therefore:
\begin{align*}
\var^S_{\Y, \Z_{\rho} } (\hat{\tau}^T_Y | \x )  &=  \EE^S_{\Y} \left[ \var_{\Z} (\hat{\tau}^T_Y | \x, \Y, \rho = 1)  | \x \right] + \frac{2}{n^2}\sigma_{\E}^2 \sum_{i=1}^n w_i^2 \\
&= \frac{4}{n^2}\beta^2 \EE_{\Z} \left[ ( \sum_{i=1}^n Z_i w_i  x_i )^2 \big | \x, \rho = 1 \right] + \frac{6}{n^2} \sigma_{\E}^2 \sum_{i=1}^n w_i^2
\end{align*}
\end{proof}
In order to prove Theorem~\ref{thm:optimality_1d}, we will show that for a random variable $U$ with $\EE[U] = 0$, among events $\Omega$ preserve the expectation $\EE[U| \Omega] = 0$, truncating the tail results in the smallest variance. Note that if $\rho(\x, \Z) = \rho(\x, -\Z)$ it follows from Lemma~\ref{lem:helper} that $\EE[\frac{2}{n}\sx^T\Z | \rho = 1] = \EE[\frac{2}{n}\sx^T\Z] = 0$. 

In order to prove Theorem 1 we show how to minimize the variance of a random variable: 
\begin{lemma}
\label{lem:var_optimality_1d}
Let $U \in \mathcal{R}$ be a random variable such that $\EE[U]=0$. 
Let $u_{\alpha}$ be such that $\PP(U^2 < u_{\alpha}) = 1 - \alpha$.
Let $\Omega$ be an event such that $\PP(\Omega ) \ge 1- \alpha$ and $\EE[U | \Omega] = 0$. Then:
\begin{align*}
     \EE(U^2 | U^2 < u_{\alpha}) \le \EE(U^2 | \Omega ) 
\end{align*}
\end{lemma}

\begin{proof}
Let $p(u)$ be the pdf of $U$. Define $f(u)$ as follow:
\begin{align*}
    f(u) =p(U = u ,  \Omega) 
\end{align*}
then:
\begin{align*}
   p(u | \Omega) = \frac{p(U=u, \Omega)}{\PP(\Omega)}  = \frac{f(u)}{1-\alpha}\;.
\end{align*}

Therefore:
\begin{align*}
\EE[U^2 | \Omega] = \int_u u^2 \frac{f(u)}{1-\alpha} du\;.
\end{align*}
We want to minimize $\EE(U^2| \Omega)$:
\begin{align*}
   \int_u u^2 \frac{f(u)}{1-\alpha} du
\end{align*}
subject to:
\begin{align*}
   & 0 \le f(u) \le p(u) ~\forall u \\
& \PP(\Omega) = \int_u f(u) du = 1 - \alpha
\end{align*}
This can be done by maximize $f(u)$ so that $f(u) = p(u)$ for the smallest $u^2$, which is equal to set $\Omega$ to be the event $U^2 < u_{\alpha}$. 
\end{proof}
 \begin{proof}[Proof of Theorem~\ref{thm:optimality_1d}]

 Let $V: = \frac{2}{n}\sum_i w_i x_i Z_i$ and $B = V\var(V)^{-1/2}$. From Lemma~\ref{lem:decomposition_d1}:
\begin{align*}
\var^S_{\Y, \Z_{\rho} } (\hat{\tau}^T_Y | \x)
    &= \beta^2 \EE_{\Z} \left [ V^2 \bigg | \x, \rho =1 \right ] + \frac{6}{n^2}\sigma^2_\E \sum_{i=1}^n w_i^2. \\
    &= \beta^2 \var(V) \EE_{\Z} \left [ B^2 \bigg | \x, \rho =1 \right ] + \frac{6}{n^2}\sigma^2_\E \sum_{i=1}^n w_i^2. 
\end{align*}
Since $\rho(\x, \Z) = \rho(\x, -\Z)$, from Lemma~\ref{lem:helper} we have $\EE_{\Z}[B | \x, \rho =1] = 0$, which satisfies the criteria in Lemma~\ref{lem:var_optimality_1d}. 

Let $\eta : = 1- \PP(\rho =1 | \x)$. Then $\eta \le \alpha$. Let $b_{\eta}$ be such that $\PP( B^2 <b_{\eta} | \x) = 1- \eta$ and $b_{\alpha}$ be such that $\PP( B^2 <b_{\alpha} | \x) = 1- \alpha$. From Lemma~\ref{lem:var_optimality_1d}:
\begin{align*}
    \EE_{\Z} \left [ B^2 \big | \x, \rho =1 \right  ] &\ge  \EE_{\Z} \left [ B^2 \big | \x, B^2 <b_{\eta} \right  ] \\
    &\ge  \EE_{\Z} \left [ B^2 \big | \x, B^2 <b_{\alpha} \right  ] \text{ because $b_{\eta} \ge b_{\alpha}$}\\
    &\ge \EE_{\Z} \left [ B^2 \big | \x, \phi^{\alpha}_T = 1 \right  ]
\end{align*}
 \end{proof}
 \begin{proof}[Proof of Corollary~\ref{col:var1d}]
Let $\rho$ being the constant function $\rho(\x,\Z)=1$ for all $\x,\Z$. Then:
\begin{align*}
   \var^S_{\Y, \Z_{\rho} } (\hat{\tau}^T_Y | \x) = \var^S_{\Y, \Z} (\hat{\tau}^T_Y | \x)
\end{align*}
From Theorem~\ref{thm:optimality_1d} we have: 
    \begin{align*}
\var^S_{\Y, \Z_{\phi^{\alpha}_T}} (\hat{\tau}^T_Y | \x) \le \var^S_{\Y, \Z_{\rho} } (\hat{\tau}^T_Y | \x) = \var^S_{\Y, \Z} (\hat{\tau}^T_Y | \x)
\end{align*}
 \end{proof}
 \section{Discussion on Section \ref{sec:asymptotic}}
 \label{apx:asymptotic}
\begin{proof}[Proof of Lemma~\ref{lem:unconditional_variance2}]
By law of total variance:
\begin{align*}
\var^S_{\X, \Y, \Z_{\rho}  } (\hat{\tau}^T_Y)  = \EE^S_{\X, \Y} \left[\var_{ \Z_{\rho} } (\hat{\tau}^T_Y | \X, \Y)\right] + \var^S_{\X, \Y} \left (\EE_{\Z_{\rho} }[\hat{\tau}^T_Y | \X, \Y] \right)
\end{align*}
Since $\rho(\x, \Z) = \rho(\x, -\Z)$, from Lemma~\ref{lem:expectation}: 
\begin{align*}
   \mathbb{E}_{\Z}[\hat{\tau}^T_Y | \X, \Y,  \rho = 1 ] 
= \frac{1}{n} \sum_{i=1}^n W_i  (Y^1_i - Y^0_i) 
\end{align*}
Therefore:
\begin{align*}
    \var^S_{\X, \Y} \left (\EE_{\Z}[\hat{\tau}^T_Y | \X, \Y,  \rho = 1] \right)
     =\var^S_{\X, \Y} \left (\frac{1}{n} \sum_{i=1}^n W_i  (Y^1_i - Y^0_i)  \right) 
\end{align*}
\end{proof}
We now prove Lemma~\ref{lem:R2}. 
We use the following result in \cite{harshaw2019balancing} to prove Lemma~\ref{lem:R2}. 
\begin{lemma}[Lemma A1 in \cite{harshaw2019balancing}]
\label{lem:decomposition1}
Let $y^*_i$ denote the observed outcome of sample $i$: 
\begin{align*}
    \frac{2}{n}(\sum_{z_i = 1}  y^*_i - \sum_{z_i = -1} y^*_i) - \frac{1}{n} \sum_{i=1}^n  (y^1_i - y^0_i) = \frac{2}{n} \cc^T \z
\end{align*}
where $c_i= \frac{y_i^1 + y_i^0}{2}$ and $\cc := (c_1, \cdots, c_n)$. 
\end{lemma}
We will also use the following lemma:
\begin{lemma}
\label{lem:Qhelper}
Let $Q:= \frac{n-1}{n}\EE [\Z\Z^T] $. Let $\I_n$ denote the $n \times n$ identity matrix and $\1$ denote the $n$ dimensional vector of $1$. Then: 
\begin{align*}
Q &= \I_n - \frac{1}{n}\1\1^T.\\
Q &= Q^T\\
Q &= Q^2 = Q^T Q = QQ^T. 
\end{align*}
Let $\s \in \mathcal{R}^{n \times d}$ be a matrix. Then  $$Q\s = \s - \text{avg}(\s)$$ where $\text{avg}(\s) \in \mathcal{R}^d$ is the average of rows of $\s$. 
\end{lemma}
\begin{proof}
First we will show that:
\begin{align*}
    \EE[\Z\Z^T] = \frac{n}{n-1} \left( \I_n - \frac{1}{n}\1\1^T \right) 
\end{align*}
by showing that $\EE[Z_i^2] = 1$ and $\EE[Z_i Z_j] = -\frac{1}{n-1}$ when $i \neq j$. 
First we have that $\EE[Z_i^2] = 1$ because $Z_i^2 = 1$. Since there are exactly $n/2$ samples with value $Z_i = 1$ and $n/2$ samples with values $Z_i = -1$, note that $(\sum_{i=1}^n Z_i)^2 = 0$ and:
\begin{align*}
    \EE[(\sum_{i=1}^n Z_i)^2] = \EE[ \sum_{i=1}^n Z_i^2] + \sum_{i \neq j} \EE[Z_i Z_j]\;.
\end{align*}
Since all pairs $(i, j)$ where $i \neq j$ have equal roles and there are $n(n-1)$ such pairs: 
\begin{align*}
    \EE[Z_i Z_j] &= \frac{\EE[(\sum_{i=1}^n Z_i)^2] - \EE[ \sum_{i=1}^n Z_i^2] }{n(n-1)} \\
    &= \frac{0- n }{n(n-1)} \\
    & = \frac{-1}{n-1}
\end{align*}
Since $Q$ is symmetric, $Q = Q^T$. We will show that $Q=Q^2$:
\begin{align*}
    Q^2 &= (\I_n - \frac{1}{n}\1\1^T) (\I_n - \frac{1}{n}\1\1^T) \\
    &= \I_n - \frac{1}{n} \1\1^T\I_n - \frac{1}{n} \I_n \1\1^T + \frac{1}{n^2} \1\1^T\1\1^T \\
    &= \I_n - \frac{1}{n}\1\1^T   = Q
\end{align*}
Since $Q = Q^T$, we have $Q = Q^2 = QQ^T = Q^TQ$. 
For the last property:
$$Q\s = \I_n \s - \frac{1}{n}\1\1^T\s = \s - \text{avg}(\s)$$
because $\I_n \s = \s$ and $\frac{1}{n}\1\1^T\s = \text{avg}(\s)$
\end{proof}
\begin{proof}[Proof of Lemma~\ref{lem:R2}]
For any matrix $\s \in \mathcal{R}^{n \times d}$ we will compute $R^2_{\s} := Corr(\hat{\tau}^T_Y, \frac{2}{n} \Z^T\s)$ where for any $Y \in \mathcal{R}, X \in \mathcal{R}^d$, $Corr(Y, X)$ is defined as: 
\begin{align*}
    Corr(Y, X) &= Corr(Y, X^T {\beta^*}^) \\
    &= \frac{Cov(Y, X^T{\beta^*})}{\sqrt{\var(Y)} \sqrt{\var(X^T{\beta^*})}}
\end{align*}
where $\beta^* = \arg\min_{\hbeta} \EE \| Y -  X^T\hbeta \|^2 $. Substituting $\s = \x$ and $\s = \sx$ will give us $R^2_{\x}$ and $R^2_{\sx}$. 

Let $\sdelta_i = \tilde{y}^1_i - \tilde{y}^0_i$ and $\sbdelta := (\sdelta_1 , \cdots, \sdelta_n)$. From Lemma~\ref{lem:decomposition1}, we have: $$\hat{\tau}^T_Y = \frac{2}{n}\Z^T\scc+ \frac{1}{n}\1^T \sbdelta$$ where $\1 \in \mathcal{R}^n$ is a vector of $1$. 

We note that conditioning on $\y$, $\1^T \sbdelta$ is a constant independent of $\Z$. Let $Q:= \frac{n-1}{n}\EE [\Z\Z^T]$ and note that $Q = Q^T$ and $Q = Q^2$. First, let us compute $\beta^* = \arg\min_{\hbeta} \EE_{\Z} \|\htau^T_Y - \frac{2}{n} \Z^T \s \hbeta \|^2  .$ We have, 
\begin{align*}
\beta^*& = \arg\min_{\hbeta} \EE_{\Z} \|\htau^T_Y - \frac{2}{n} \Z^T \s \hbeta  \|^2 \\
    &=\arg \min_{\hbeta} \EE_{\Z} \|  \frac{2}{n}  \Z^T\scc +   \frac{1}{n}  \1^T \sbdelta -   \frac{2}{n} \Z^T \s \hbeta \|^2 \\
    &= \arg \min_{\hbeta} \EE_{\Z} \|\Z^T\scc  - \Z^T \s \hbeta  \|^2 \\
    &= \arg \min_{\hbeta} (\scc - \s \hbeta)^T \EE[\Z\Z^T] (\scc - \s \hbeta) \\
    &= \arg \min_{\hbeta} (\scc - \s \hbeta)^T Q (\scc - \s \hbeta) \\
    &= \arg \min_{\hbeta} (\scc - \s \hbeta)^T Q^T Q (\scc - \s \hbeta) \\
    &=\arg \min_{\hbeta} \|Q\scc - Q\s \beta\|^2 .
\end{align*}
Using the fact that $Q = Q^T Q$, we have $\beta^* = (\s^T Q \s)^{-1} \s^T Q  \scc.$ 
By definition, we have
\begin{align*}
     Corr(\hat{\tau}^T_Y, \frac{2}{n}\Z^T \s) &= \frac{ \EE_{\Z} \left[\hat{\tau}^T_Y  \frac{2}{n}\Z^T \s \beta^*\right] -  \EE_{\Z} \left[\hat{\tau}^T_Y \right] \EE_{\Z} \left[ \frac{2}{n}\Z^T \s \beta^*\right]}{\sqrt{\var_{\Z}(\hat{\tau}^T_Y) \var_{\Z}(\frac{2}{n}\Z^T \s  \beta^*)}} \\
    &=  \frac{ \EE_{\Z} \left[\hat{\tau}^T_Y  \Z^T \s \beta^*\right]}{\sqrt{\var_{\Z}(\hat{\tau}^T_Y) \var_{\Z}(\Z^T \s  \beta^*)}} \text{ because $\EE[\Z]=0$}\\
    &=  \frac{ \EE_{\Z} \left[\left( \frac{2}{n}\scc^T\Z+ \frac{1}{n}\1^T \sbdelta\right)  \Z^T \s \beta^*\right]}{\sqrt{\var_{\Z}\left( \frac{2}{n}\Z^T\scc+ \frac{1}{n}\1^T \sbdelta\right)  \var_{\Z}(\Z^T \s  \beta^*)}} \\
    &=  \frac{ \EE_{\Z} \left[\left( \frac{2}{n}\scc^T\Z \right)  \Z^T \s \beta^*\right]}{\sqrt{\var_{\Z}\left( \frac{2}{n}\Z^T\scc\right)  \var_{\Z}(\Z^T \s  \beta^*)}} \\
      &=  \frac{ \EE_{\Z} \left[\scc^T\Z  \Z^T \s \beta^*\right]}{\sqrt{\var_{\Z}\left( \Z^T\scc\right)  \var_{\Z}(\Z^T \s  \beta^*)}} 
\end{align*}
For the numerator we have:
\begin{align*}
    \EE_{\Z} \left[\scc^T\Z  \Z^T \s \beta^*\right] &= \scc^T Q \s \beta^* \\
    &=\frac{n}{n-1} \scc^T Q \s (\s^T Q \s)^{-1} \s^T Q  \scc \\
     &= \frac{n}{n-1}\scc^T Q \s (\s^T Q \s)^{-1} \s^T Q \s (\s^T Q \s)^{-1}\s^T Q  \scc \\
     &= \frac{n}{n-1}\left(\scc^T Q \s (\s^T Q \s)^{-1} \s^T Q \right) \left(Q \s (\s^T Q \s)^{-1}\s^T Q  \scc \right) \\
     &= \frac{n}{n-1}({\beta^*}^T \s^TQ ) (Q\s \beta^*) \\
     & = \frac{n}{n-1} \| Q\s\beta^*\|^2
\end{align*}
Let $u = Q\s\beta$ and $v = Q\scc -  Q\s\beta$. We will show that $u$ and $v$ are orthogonal, therefore $\|Q\s\beta^*\|^2 = \|Q\scc\|^2 - \|Q\scc -  Q\s\beta\|^2$: 
\begin{align*}
    u^T v &=  (Q\scc -  Q\s\beta^*)^T (Q\s\beta^*) \\
    &= \scc^T Q\s\beta^* - {\beta^*}^T \s^T Q \s\beta^* \\
    &= \scc^T Q\s\beta^* - \| Q\s\beta^*\|^2\\
    &= 0 \;.
\end{align*}
Therefore $\|Q\s\beta^*\|^2 = \|Q\scc\|^2 - \|Q\scc -  Q\s\beta\|^2$. 

For the denominator, since $\EE[\Z]=0$ we have:
\begin{align*}
    \var_{\Z}\left( \Z^T\scc\right)  \var_{\Z}(\Z^T \s  \beta^*) &= \EE_{\Z} [\scc^T\Z\Z^T\scc] \EE_{\Z}[{\beta^*}^T \s^T \Z\Z^T \s\beta^*] \\
    &= \frac{n^2}{(n-1)^2}(\scc^T Q \scc ) ({\beta^*}^T \s^T Q \s\beta^*)\\
    &= \frac{n^2}{(n-1)^2}(\scc^T Q^T Q \scc ) ({\beta^*}^T \s^T Q^TQ \s\beta^*) \\
    &= \frac{n^2}{(n-1)^2}\|Q\scc\|^2 \|Q \s\beta^*\|^2
\end{align*}
Putting the numerator and denominator together we have:
\begin{align*}
    R^2_{\s} &= Corr(\hat{\tau}^T_Y, \frac{2}{n}\Z^T \s) \\
    &=\frac{\| Q\s\beta^*\|^2}{\|Q\scc\| \|Q \s\beta^*\|} \\
    &=\frac{\| Q\s\beta^*\|}{\|Q\scc\| } \\
    &=\frac{\sqrt{\|Q\scc\|^2 - \|Q\scc -  Q\s\beta\|^2}}{\|Q\scc\| }
\end{align*}
Substituting $\s = \x$ and $\s = \sx$ gives us the expression for $R^2_{\x}$ and $R^2_{\sx}$. 
\end{proof}

\begin{proof}[Proof of Theorem~\ref{thm:asymptotic_optimal}]

We have
\begin{align*}
    \tilde{C} = \tilde{X}^T \beta + \tilde{\E}
\end{align*}
where $C = \frac{Y^0+Y^1}{2}$,  $\E = \frac{\E_0+\E_1}{2}$, $\beta = \frac{\beta_0+\beta_1}{2}$, $\tilde{C} = \frac{p_T(X)}{p_S(X)} C, \tilde{X} = \frac{p_T(X)}{p_S(X)} X$ and $\tilde{\E} = \frac{p_T(X)}{p_S(X)} \E$.
Since $Y_i$, $X_i$ and $W_i$ have finite $8$th moment, $\tilde{C}_i$ and $\tilde{X}_i$ have finite $4$th moment using Cauchy-Schwartz inequality.
Let $S \in \mathcal{R}^d$ be a random variable independent of $\E_i$ and with finite $4$th moment. Let $\SL \in \mathcal{R}^{n \times d}$ be $n$ samples $S_1, \cdots, S_n$ of $S$. By the definition of $R^2$, 
\begin{align*}
 R_{\SL}^2 = \frac{\|Q\SC \|^2 - \min_{\hbeta} \|Q\SC - Q\SL \hbeta \|^2}{\|Q\SC\|^2} .
\end{align*}

We will show that $\lim_{n\to \infty}R_{\tilde{\X}}^2 \geq \lim_{n\to \infty}R_{\SL}^2$ almost surely for any $S$. It is sufficient to show $\lim \min_{\hbeta} \|Q\SC - Q\SX \hbeta \|^2 \leq \lim \min_{\hbeta} \|Q\SC - Q\SL \hbeta \|^2$ almost surely. From Lemma~\ref{lem:Qhelper}, note that for any matrix $\s \in \mathcal{R}^{n \times d}$ with $n$ rows,  $\frac{n-1}{n}Q\s = \s - \text{avg}(\s)$ where $\text{avg}(\s) \in \mathcal{R}^d$ is the average of rows of $\s$. Let $\beta^* = \arg\min_{\hbeta} \lim_{n \rightarrow \infty}   \frac{1}{n} \|Q\SC - Q\SL \hbeta \|^2$ and $\tilde{\beta}   = \arg\min_{\hbeta}  \frac{1}{n} \|Q\SC - Q\SL \hbeta \|^2$. If $S_i$ and $\tilde{C}_i$ have finite $4$th moment, by strong law of large number $\lim_{n \rightarrow \infty} \tbeta = \beta^*$ almost surely. 
We have: 
\begin{align*}
    &\frac{1}{n} \lim_{n \rightarrow \infty}   \min_{\hbeta} \|Q\SC - Q\SL \hbeta \|^2 \\ 
    & = \frac{1}{n} \lim_{n \rightarrow \infty}    \|Q\SC - Q\SL \beta^* \|^2 + 2 (Q\SC - Q\SL \beta^*)^T ( Q\SL \beta^* - Q\SL \tbeta) + \norm{Q\SL \beta^* - Q\SL \tbeta}^2 \\ 
    &= \frac{1}{n} \lim_{n \rightarrow \infty}    \|Q\SC - Q\SL \beta^* \|^2 + 2 \lim_{n \rightarrow \infty} (Q\SC - Q\SL \beta^*)^T Q\SL \lim_{n \rightarrow \infty} (\beta^* -  \tbeta) + \lim_{n \rightarrow \infty} (\beta^* -  \tbeta)^T \lim_{n \rightarrow \infty} \SL^TQ\SL \lim_{n \rightarrow \infty} (\beta^* -  \tbeta) \\&\text{ because $\tilde{C}_i$ and $S_i$ having finite $4$th moment implies $\lim_{n \rightarrow \infty} (Q\SC - Q\SL \beta^*)^T Q\SL$ and $\lim_{n \rightarrow \infty} \SL^TQ\SL$ are finite} \\
    &= \lim_{n \rightarrow \infty}   \frac{1}{n} \|Q\SC - Q\SL \beta^* \|^2 \text{ almost surely } \\
    &= \min_{\hbeta}\lim_{n \rightarrow \infty}   \frac{1}{n} \|Q\SC - Q\SL \hbeta \|^2 \\
     & =   \min_{\hbeta}\lim_{n \rightarrow \infty}   \frac{1}{n} \frac{n^2}{(n-1)^2} \norm{\frac{n-1}{n}Q\SC - \frac{n-1}{n}Q\SL \hbeta }^2 \\
    & =   \min_{\hbeta}\lim_{n \rightarrow \infty}   \frac{1}{n} \|(\SC - \SL\hbeta) - (\text{avg}(\SC) - \text{avg}(\SL) \hbeta )\|^2  \\
         & = \min_{\hbeta} \var(\TC - S^T\hbeta) \text{ almost surely if $\tilde{C}_i$ and $S_i$ have finite $4$th moment} \\
         &= \min_{\hbeta} \EE[(\TC - S^T\hbeta)^2] - \left(\EE[\TC - S^T\hbeta]\right)^2 \\
         &= \min_{\hbeta}   \EE[ \tilde{X}^T\beta - S^T \hbeta ]^2 + \EE[\SE^2] - \left(\EE[\tilde{X}^T\beta - S^T\hbeta]\right)^2 \text{ because } \EE[\SE] =0 \text{ and $\E$ is independent of $\tilde{X}$ and $S$}\\
         &= \min_{\hbeta} \var(\tilde{X}^T\beta - S^T\hbeta)+ \EE[\SE^2]
         \ge \EE[\SE^2]
\end{align*}

When $S =\tilde{X},$ this is minimized, therefore:
\begin{align*}
    \lim_{n \rightarrow \infty} R^2_{\SX} \ge \lim_{n \rightarrow \infty} R^2_{\SL} \text{ almost surely.} 
\end{align*}
 Substituting $\SL = \X$: 
\begin{align*}
    \lim_{n \rightarrow \infty} R^2_{\SX} \ge \lim_{n \rightarrow \infty} R^2_{\X} \text{ almost surely.} 
\end{align*}
 Recall that: 
\begin{align*}
& \text{as-var}_{\Z} \left( \hat{\tau}^T_Y  | \x, \y,  M\left(\frac{2}{n}\Z^T\s\right) \leq a \right) = \lim_{n \rightarrow \infty} \var(\hat{\tau}^T_Y | \x, \y) (1 - (1-v_{d,a}) R_{\s}^2 ),
 \end{align*}
 where $\text{as-var}$ is the variance of the asymptotic sampling distribution. Let $s(a)$ denote the rejection probability $\PP(\phi_S = 0 |\x) = 0$ when using threshold $a$ in Source Balance, and $t(a)$ denote the rejection probability $\PP(\phi_T = 0 |\x) = 0$ when using threshold $a$ in Target Balance. 
 We have:
 \begin{align*}
\text{as-var}_{\Z} \left( \hat{\tau}^T_Y  | \X, \Y,  \phi^{s(a)}_S =1 \right) &=
    \text{as-var}_{\Z} \left( \hat{\tau}^T_Y  | \X, \Y,  M\left(\frac{2}{n}\Z^T\X\right) \leq a \right) \\&= \lim_{n \rightarrow \infty} \var(\hat{\tau}^T_Y | \X, \Y) (1 - (1-v_{d,a}) R_{\X}^2 ) \\ &\ge \lim_{n \rightarrow \infty} \var(\hat{\tau}^T_Y | \X, \Y) (1 - (1-v_{d,a}) R_{\SX}^2 ) \text{ almost surely} \\
     &=  \text{as-var}_{\Z}\left( \hat{\tau}^T_Y  | \X, \Y,  M\left(\frac{2}{n}\Z^T\SX \right) \leq a \right) \\
     &=  \text{as-var}_{\Z}\left( \hat{\tau}^T_Y  | \X, \Y,  \phi^{t(a)}_T = 1 \right)
 \end{align*}
Now we will show that for any $\x$ and $\sx$, $\lim_{n \rightarrow \infty} s(a) = \lim_{n \rightarrow \infty} t(a)$. Let $U \in \mathcal{R}^d$ be a standard multivariate random variable. We have:
\begin{align*}
    \lim_{n \rightarrow \infty} s(a) 
    &= \lim_{n \rightarrow \infty} \PP( M\left(\frac{2}{n}\Z^T\x\right) \leq a) \\
    &= \lim_{n \rightarrow \infty}\PP( \|B_S\|^2 < a ) \text{ where } B_S = \frac{2}{n}\Z^T\x Cov(\frac{2}{n}\Z^T\x)^{-1/2} \\
    &= \PP( \| U \|^2 < a ) \text{ because $B_S$ converges in distribution to $U$ by finite central limit theorem} 
\end{align*}
Similarly we have:
\begin{align*}
    \lim_{n \rightarrow \infty} t(a) 
    &= \lim_{n \rightarrow \infty} \PP( M\left(\frac{2}{n}\Z^T\sx\right) \leq a) \\
    &= \lim_{n \rightarrow \infty}\PP( \|B_T\|^2 < a ) \text{ where } B_T := \frac{2}{n}\Z^T\sx Cov(\frac{2}{n}\Z^T\sx)^{-1/2} \\
    &= \PP( \| U \|^2 < a ) \text{ because $B_T$ converges in distribution to $U$ by finite central limit theorem} 
\end{align*}
Therefore $\lim_{n \rightarrow \infty} t(a) = \lim_{n \rightarrow \infty} s(a)$. When the sample size is large, with the same rejection probability, using Target Balance results in a smaller asymptotic variance than Source Balance . 
\end{proof}

\section{Proofs of Section \ref{apx:additional_results}}
\label{apx:additional_results_proofs}

In this Section we present the proof of  Lemma~\ref{lem:variance_decomposition}, Lemma~\ref{lem:expected_beta}, Lemma~\ref{lem:trace_optimality}, Theorem~\ref{thm:expected_beta_optimal} and Corollary~\ref{col:var_d_greater_1}. 

In order to prove Lemma~\ref{lem:variance_decomposition}, we first prove the following lemma. 
\begin{lemma}[minor changes to Lemma 1 in \citep{harshaw2019balancing}]
\label{lem:variance_decomposition_condition_on_y}
Let $\sepsilon_i = \tilde{c}_i - \beta^T \tilde{x}_i$ and $\se = (\sepsilon_1 , \cdots, \sepsilon_n)$. For any function $\rho(\x, \Z) \in \{0,1\}$ satisfying  $\rho(\x, \Z) = \rho(\x, -\Z)$:
\begin{align}
\frac{n^2}{4} \var_{\Z} (\hat{\tau}^T_Y |  \x, \y, \rho = 1) &=  Cov(\scc^T \Z| \rho = 1) \label{eq:var_decomp1}\\
&= \beta^T  Cov(\sx^T\Z | \rho = 1]  \beta +   Cov(\se^T \Z | \rho = 1)    + 2 \beta^T Cov( \sx^T\Z, \se^T\Z | \rho = 1) \label{eq:var_decomp2}
    \end{align}
\end{lemma}

\begin{proof}[Proof of Lemma~\ref{lem:variance_decomposition_condition_on_y}]
By definition:
\begin{align*}
\var_{\Z}(\hat{\tau}^T_Y |  \x, \y, \rho = 1) 
&= \EE_{\Z} \left[(\hat{\tau}^T_Y - \mathbb{E}_{\Z}[\hat{\tau}^T_Y | \x, \y, \rho = 1 ] )^2 | \x, \y, \rho = 1 \right] 
\end{align*}
We have:
\begin{align*}
\EE_{\Z} [\hat{\tau}^T_Y | \x , \y, \rho = 1] &= \frac{2}{n}\EE_{\Z} \left[ \sum_{Z_i = 1} w_i y^*_i - \sum_{Z_i = -1} w_i y^*_i \bigg| \rho = 1 \right] \\
&= \frac{2}{n}\EE \left[ \sum_{i=1}^n A_i w_i y^1_i  - \sum_{i=1}^n (1-A_i) w_i y^0_i   \bigg| \rho = 1 \right]\\
&= \frac{2}{n} \left( \sum_{i=1}^n \EE[A_i| \rho = 1] w_i y^1_i - \sum_{i=1}^n \EE[1-A_i| \rho=1] w_i y^0_i \right)\\
&= \frac{1}{n} \left(\sum_{i=1}^n  w_i y^1_i - \sum_{i=1}^n w_i y^0_i \right) \text{ because $\EE[A_i| \rho = 1] = 1/2$ by Lemma~\ref{lem:helper}}
\end{align*}
Therefore using Lemma~\ref{lem:decomposition1}:
\begin{align*}
\var_{\Z}(\hat{\tau}^T_Y |  \x, \y, \rho = 1)  &= \EE_{\Z} \left[ \left( \frac{2}{n}(\sum_{Z_i = 1} w_i y^*_i - \sum_{Z_i = -1} w_i y^*_i) - \frac{1}{n} \sum_{i=1}^n w_i (y^1_i - y^0_i) \right)^2 \bigg|\x, \y, \rho = 1 \right] \\
&= \frac{4}{n^2}  \EE[\scc^T \Z\Z^T \scc| \x, \y, \rho = 1]\\
&= \frac{4}{n^2} Cov(\scc^T \Z | \x, \y, \rho = 1) \text{ because } \EE[\scc^T \Z | \x, \y, \rho =1] = 0 \text{ from Lemma~\ref{lem:helper}} \\
&=  \frac{4}{n^2} Cov((\sx\beta + \se)^T \Z | \x, \y, \rho = 1) \\
&= \beta^T  Cov(\sx^T\Z | \x, \y, \rho = 1]  \beta +   Cov(\se^T \Z | \x, \y, \rho = 1)    + 2 \beta^T Cov( \sx^T\Z, \se^T\Z | \x, \y, \rho = 1)
\end{align*}
\end{proof}

\begin{proof}[Proof of Lemma~\ref{lem:variance_decomposition}]
By law of total variance:
\begin{align*}
\var^S_{\Y, \Z_{\rho} } (\hat{\tau}^T_Y | \x )  
&= \EE^S_{\Y} \left[ \var_{\Z} (\hat{\tau}^T_Y | \x, \Y, \rho = 1)  | \x \right] + \var^S_{\Y} ( \EE_{\Z} [\hat{\tau}^T_Y | \x, \Y, \rho = 1] |\x) \\
&=  \EE^S_{\Y} \left[ \var_{\Z} (\hat{\tau}^T_Y | \x, \Y, \rho = 1)  | \x \right]  + \var^S_{\Y} \left( \frac{1}{n} \sum_{i=1}^n w_i (Y^1_i  - Y^0_i) | \x  \right) \\
&=  \EE^S_{\Y} \left[ \var_{\Z} (\hat{\tau}^T_Y | \x, \Y, \rho = 1)  | \x \right]  + \frac{1}{n^2} \sum_{i=1}^n w_i^2 \var (\E_{1}  - \E_{0})  \\
&=  \EE^S_{\Y} \left[ \var_{\Z} (\hat{\tau}^T_Y | \x, \Y, \rho = 1)  | \x \right]  + \frac{2}{n^2}\sigma_{\E}^2 \sum_{i=1}^n w_i^2 
\end{align*}
From Lemma~\ref{lem:variance_decomposition_condition_on_y}:
\begin{align*}
   \frac{n^2}{4} \var_{\Z} (\hat{\tau}^T_Y |  \x, \y, \rho = 1) &=  \beta^T  Cov(\sx^T\Z |\x, \y, \rho = 1]  \beta +   Cov(\se^T \Z | \x, \y,\rho = 1)    + 2 \beta^T Cov( \sx^T\Z, \se^T\Z |\x, \y, \rho = 1) \\
   &= \beta^T  Cov(\sx^T\Z | \x, \y,\rho = 1]  \beta +   \se^T Cov(\Z | \x, \y,\rho = 1) \se   + 2 \beta^T Cov( \sx^T\Z, \Z | \x, \y,\rho = 1) \se
\end{align*}
Recall that $Y^1_i = \beta_1^T X_i + \E^1_i$ and $Y^0_i = \beta_1^T X_i + \E^0_i$. 
Let $\E_i = \frac{\E^1_i + \E^0_i}{2}$ and $\SBE = (\E_1, \cdots, \E_n)$. Since $\se$ is the value of $\SBE$ we have: 
\begin{align*}
    &\frac{n^2}{4} \EE^S_{\Y} \left[ \var_{\Z} (\hat{\tau}^T_Y | \x, \Y, \rho = 1)  | \x\right]  \\
    &=   \beta^T  Cov(\sx^T\Z | \x, \y,\rho = 1]  \beta +   \EE^S_{\Y} [\SBE^T Cov(\Z | \x, \Y,\rho = 1) \SBE | \x]   + 2 \beta^T Cov( \sx^T\Z, \Z | \x, \y, \rho = 1) \EE[\SBE | \x] \\
    &=\beta^T  Cov(\sx^T\Z | \x, \rho = 1]  \beta +   \EE^S_{\Y} [ Cov(\SBE^T\Z | \x, \rho = 1) | \x] \text{ because }\EE[\SBE | \x] = 0 
\end{align*}

The second term: 
\begin{align*}
      &\EE^S_{\Y} [ Cov(\SBE^T\Z |\x,  \rho = 1) |  \x] \\
      &= \EE_{\SBE} [ \EE_{\Z}[\SBE^T\Z \Z^T \SBE | \x,\rho = 1]  |  \x] \\
     &= \EE_{\SBE}\left[\sum_{i=1}^n \sum_{j=1}^n \EE_{\Z}[w_i \E_i Z_i Z_j \E_j w_j | \rho =1] |\x\right]\\
     &= \EE_{\SBE}\left[\sum_{i=1}^n \EE[w_i^2 \E_i^2 Z_i^2  | \rho = 1] + \sum_{i \neq j} \EE[w_i \E_i Z_i Z_j \E_j w_j | \rho = 1] | \x \right]\\
     &= \EE_{\SBE}\left[ \sum_{i=1}^n \EE[w_i^2 \E_i^2  1 | \rho = 1] + \sum_{i \neq j} \EE[w_i  Z_i Z_j \E_j w_j | \rho = 1] \EE[\E_i | \rho=1]| \x \right] \text{ because }Z_i^2=1\\
     &= \sum_{i=1}^n \EE[w_i^2 \E_i^2  ] \text{ because }  \EE[\E_i | \rho =1]=0\\
     &= \sum_{i=1}^n w_i^2 \sigma_{\E}^2 
\end{align*}
Putting all together:
\begin{align*}
   \var^S_{\Y, \Z_{\rho} } (\hat{\tau}^T_Y | \x) &= \frac{4}{n^2} \left( \beta^T  Cov(\sx^T\Z | \x, \rho = 1]  \beta + \sum_{i=1}^n w_i^2\sigma^2_{\E} \right) + + \frac{2}{n^2}\sigma_{\E}^2 \sum_{i=1}^n w_i^2  \\
    &= \frac{4}{n^2} \beta^T  Cov_{\Z}(\sx^T\Z | \x, \rho = 1]  \beta + \frac{6}{n^2}\sigma_{\E}^2 \sum_{i=1}^n w_i^2  \\
\end{align*}
\end{proof}
\begin{proof}[Proof of Lemma~\ref{lem:expected_beta}]
We use the same decomposition of $\beta^T Cov_{\Z}(V|\x, \Omega) \beta$ as in~\citep{harshaw2019balancing}.  Let $\e_1, ..., \e_n$ and $\lambda_1, .., \lambda_n$ be the normalized eigenvectors and corresponding eigenvalues of  matrix $Cov_{\Z}(V | \x, \E) $. Since $Cov_{\Z}(V | \x, \E)$ is symmetric, the eigenvectors form an orthonormal basis so we can write $\beta$ as a linear combination of $\e_1, .., \e_n$ and get:
\begin{align*}
   \beta = \norm{\beta}\sum_{i=1}^n \eta_i \e_i 
\end{align*}
where $\eta_i = \langle \beta, \e_i \rangle/\norm{\beta}$ is the coefficient that captures the alignment of the weighted outcome $\beta$ with respect to the eigenvector $\e_i$. 
Therefore:
\begin{align*}
   \beta^T Cov_{\Z}(V | \x, \Omega) \beta = \norm{\beta}^2\sum_{i=1}^n \eta_i^2 \lambda_i
\end{align*}
Then:
\begin{align*}
  \EE_{\beta} \left[\beta^T Cov_{\Z}(V |\x, \Omega) \beta \right]&= \EE_{\beta} \left[\norm{\beta}^2\sum_{i=1}^n \eta_i^2 \lambda_i \right]\\
  &= l^2 \sum_{i=1}^n \lambda_i \EE_{\beta} [\eta_i^2] \\
  &= l^2 \sum_{i=1}^n \lambda_i \EE_{\theta} cos^2(\theta) \;\parbox{30em}{ where $\theta$  is the angle between $\beta$ and $\e_i$. Since $\beta$ points to any direction with equal probability, $\theta$ is uniformly distributed in $[0,2\pi]$.} \\
  &= \frac{l^2}{2} \sum_{i=1}^n \lambda_i \\
  &= \frac{l^2}{2} \text{Trace}(Cov_{\Z}(V |\x, \Omega) ).
\end{align*}
\end{proof}
\begin{proof}[Proof of Lemma~\ref{lem:trace_optimality}]
Let $p(u)$ be the pdf of $U$. Define $f(u)$ as follow:
\begin{align*}
    f(u) =p(U = u ,  \Omega) 
\end{align*}
Then:
\begin{align*}
   p(u | \Omega) = \frac{p(U=u, \Omega)}{\PP(\Omega)}  = \frac{f(u)}{1-\alpha}
\end{align*}
Since $\PP(\Omega) = 1-\alpha$ we have: 
\begin{align*}
    \int_u f(u) du = 1 - \alpha
\end{align*}
We have:
\begin{align*}
  Trace( Cov(U | \Omega)) = Trace( \EE[UU^T | \Omega])  = Trace( \EE[UU^T | \Omega] = Trace( \EE[U^TU | \Omega] = \int_u u^Tu \frac{f(u)}{1-\alpha} du
\end{align*}
We want to minimize $Trace( Cov(U | \Omega)) $:
\begin{align*}
   \int_u u^Tu \frac{f(u)}{1-\alpha} du
\end{align*}
subject to:
\begin{align*}
   & 0 \le f(u) \le p(u) ~\forall u \\
& \int_u f(u) du = 1 - \alpha
\end{align*}
This can be done by maximize $f(u)$ so that $f(u) = p(u)$ for the smallest $u^T u$, which is equal to set $\Omega$ to be the event $\norm{U}^2 < u_{\alpha}$. 
\end{proof}
\begin{proof}[Proof of Theorem~\ref{thm:expected_beta_optimal}]
Let $\eta : = 1- \PP(\rho =1 | \x)$. Then $\eta \le \alpha$. Let $v_{\eta}$ be such that $\PP( \|V\|^2 <v_{\eta} | \x) = 1- \eta$. 
From Lemma~\ref{lem:variance_decomposition}:
\begin{align}
    \EE_{\beta} \var^S_{\Y, \Z_{\rho} }(\hat{\tau}^T_Y | \x) &= \frac{4}{n^2}  \EE_{\beta} \beta^T Cov(V|\x, \rho = 1) \beta  + \frac{6}{n^2}\sigma^2 \sum_{i=1}^n w_i^2 \\
    &= \frac{4}{n^2}  \frac{l^2}{2} \text{Trace}(Cov(V |\x, \rho = 1) )  + \frac{6}{n^2}\sigma^2 \sum_{i=1}^n w_i^2 \\
    &\ge \frac{4}{n^2}  \frac{l^2}{2} \text{Trace}(Cov(V |\x, \|V\|^2 < v_{\eta}) )  + \frac{6}{n^2}\sigma^2 \sum_{i=1}^n w_i^2 \\
    &\ge \frac{4}{n^2}  \frac{l^2}{2} \text{Trace}(Cov(V |\x, \|V\|^2 < v_{\alpha}) )  + \frac{6}{n^2}\sigma^2 \sum_{i=1}^n w_i^2 \text{ because $v_{\eta} \ge v_{\alpha}$}\\
    &\ge \frac{4}{n^2}  \frac{l^2}{2} \text{Trace}(Cov(V |\x, {\phi^{\alpha}_T}' = 1) )  + \frac{6}{n^2}\sigma^2 \sum_{i=1}^n w_i^2 \\
    &\ge \frac{4}{n^2}  \EE_{\beta} \beta^T Cov(V|\x, {\phi^{\alpha}_T}' = 1) \beta  + \frac{6}{n^2}\sigma^2 \sum_{i=1}^n w_i^2 \\
    &\ge  \EE_{\beta} \var^S_{\Y, \Z_{{\phi^{\alpha}_T}'}  }(\hat{\tau}^T_Y | \x )
\end{align}

\end{proof}
\begin{proof}[Proof of Corollary~\ref{col:var_d_greater_1}]
Let $\rho$ being the constant function $\rho(\x,\Z)=1$ for all $\x,\Z$. Then:
\begin{align*}
    \var^S_{ \Z_{\rho}, \Y} (\hat{\tau}^T_Y | \x) =  \var^S_{ \Z, \Y} (\hat{\tau}^T_Y | \x)
\end{align*}
From Theorem~\ref{thm:expected_beta_optimal} we have: 
\begin{align*}
     \EE_{\beta} \var^S_{ \Z_{\phi'_T} , \Y} (\hat{\tau}^T_Y | \x) \le \EE_{\beta} \var^S_{ \Z_{\rho}, \Y} (\hat{\tau}^T_Y | \x) = \EE_{\beta} \var^S_{ \Z, \Y} (\hat{\tau}^T_Y | \x)
\end{align*}
 \end{proof}

\end{document}